\documentclass[11pt,reqno]{amsart}
\usepackage{amssymb}
\usepackage{amsfonts}
\usepackage{amsmath}
\usepackage{stmaryrd}
\usepackage{physics}
\usepackage{braket}
\usepackage{dsfont}
\usepackage{bbold}
\usepackage{graphicx}
\usepackage{relsize}
\usepackage{makecell}
\usepackage[table,xcdraw]{xcolor}
\usepackage{enumerate}
\usepackage[pagebackref, colorlinks = true, linkcolor = blue, urlcolor  = blue, citecolor = red]{hyperref}
\usepackage[margin=1in]{geometry}
\usepackage{enumitem}
\usepackage{mathtools}
\usepackage{graphbox}
\usepackage{comment}
\usepackage{float}
\usepackage[font=scriptsize]{caption}

\newcommand{\iden}{\mathbb{1}}
\renewcommand{\epsilon}{\varepsilon}
\renewcommand{\phi}{\varphi}
\newcommand{\overbar}[1]{\mkern 1.5mu\overline{\mkern-1.5mu#1\mkern-1.5mu}\mkern 1.5mu}

\newcommand{\CDUC}{\mathsf{CDUC}}
\newcommand{\DUC}{\mathsf{DUC}}
\newcommand{\DOC}{\mathsf{DOC}}

\newcommand{\MLDUI}[1]{\mathcal{M}_{#1}(\mathbb{C})^{\times 2}_{\mathbb{C}^{#1}}}
\newcommand{\MLDOI}[1]{\mathcal{M}_{#1}(\mathbb{C})^{\times 3}_{\mathbb{C}^{#1}}}
\newcommand{\M}[1]{\mathcal{M}_{#1}(\mathbb{C})}
\newcommand{\St}[1]{\mathcal{S}_{#1}(\mathbb{C})}
\newcommand{\C}[1]{\mathbb{C}^{#1}}

\newcommand{\T}[1]{\mathcal{T}_{#1}(\mathbb{C})}

\newcommand{\U}[1]{\mathcal{U}(#1)}
\newcommand{\UU}[1]{\mathcal{U}(#1\otimes #1)}

\newcommand{\specialcell}[2][l]{%
  \begin{tabular}[#1]{@{}l@{}}#2\end{tabular}}
\newcommand{\centered}[1]{\begin{tabular}{@{}l@{}} #1 \end{tabular}}

\newtheorem{theorem}{Theorem}[section]
\newtheorem{definition}[theorem]{Definition}
\newtheorem*{definition*}{Definition}
\newtheorem{proposition}[theorem]{Proposition}
\newtheorem{corollary}[theorem]{Corollary}
\newtheorem{lemma}[theorem]{Lemma}
\newtheorem{remark}[theorem]{Remark}

\newtheorem*{conjecture*}{Conjecture}

\newcommand\vertarrowbox[3][6ex]{%
  \begin{array}[t]{@{}c@{}} #2 \\
  \left\uparrow\vcenter{\hrule height #1}\right.\kern-\nulldelimiterspace\\
  \makebox[0pt]{\scriptsize#3}
  \end{array}%
}

\theoremstyle{definition}
\newtheorem{example}[theorem]{Example}

\definecolor{darkgreen}{rgb}{0,0.392,0}

\author{Satvik Singh}
\email{satviksingh2@gmail.com}
\address{\parbox{\linewidth}{Department of Applied Mathematics and Theoretical Physics, \\ University of Cambridge, Cambridge, United Kingdom }}

\author{Nilanjana Datta}
\email{n.datta@damtp.cam.ac.uk}
\address{\parbox{\linewidth}{Department of Applied Mathematics and Theoretical Physics, \\ University of Cambridge, Cambridge, United Kingdom}}

\author{Ion Nechita}
\email{nechita@irsamc.ups-tlse.fr}
\address{Laboratoire de Physique Th\'eorique, Universit\'e de Toulouse, CNRS, UPS, France}

\title{Ergodic theory of diagonal orthogonal covariant quantum channels}

\begin{document}
\begin{abstract}

We analyze the ergodic properties of quantum channels that are covariant with respect to diagonal orthogonal transformations. We prove that the ergodic behaviour of a channel in this class is essentially governed by a classical stochastic matrix. This allows us to exploit tools from classical ergodic theory to study quantum ergodicity of such channels. As an application of our analysis, we study dual unitary brickwork circuits which have recently been proposed as minimal models of quantum chaos in many-body systems. Upon imposing a local diagonal orthogonal invariance symmetry on these circuits, the long-term behaviour of spatio-temporal correlations between local observables in such circuits is completely determined by the ergodic properties of a channel that is covariant under diagonal orthogonal transformations. We utilize this fact to show that such symmetric dual unitary circuits exhibit a rich variety of ergodic behaviours, thus emphasizing their importance. 
\end{abstract}

\maketitle

\tableofcontents

\section{Introduction}
Ergodicity and mixing are notions concerning the long-term qualitative behaviour of dynamical systems or stochastic processes, and their study is the subject of ergodic theory. The word ``ergodic'' was coined by Boltzmann and the study of ergodicity originated in his works on Hamiltonian flows in Statistical Mechanics. Boltzmann's ``ergodic hypothesis"~\cite{Boltzmann1871} states that, generically\footnote{That is, besides exceptions due to symmetries and integrability.} for a closed, classical Hamiltonian system, a point in phase space evolves in time and eventually visits all other points with the same energy. In other words, trajectories in phase space fills out entire energy surfaces. This hypothesis led Boltzmann to conclude that the long-term
average of an observable, as the system evolves, would be equal
to its microcanonical ensemble average. However, this hypothesis can only be true if the phase space is finite and discrete. In 1911, Ehrenfest and Ehrenfest (see e.g.~\cite{Ehrenfests1911}) introduced a modification of it, the so-called 
``quasi-ergodic hypothesis",
which states that each trajectory in the phase space, instead of filling out the entire energy surface, is generically {\em{dense}} in the
energy surface. Subsequently, von Neumann~\cite{vonNeumann1932} and Birkhoff~\cite{Birkhoff1931} laid the mathematical foundations of classical ergodic theory by establishing the mean ergodic theorem and the pointwise ergodic theorem, respectively. 

The modern notion of ergodicity (see e.g.~\cite{Birkhoff1924}) is most generally described in the setting of a measure-preserving dynamical system\footnote{Henceforth, we restrict attention to dynamical systems which are measure-preserving but drop the term `measure-preserving' for simplicity.}. An ergodic dynamical system indeed satisfies the property that Boltzmann first intended to deduce from his
hypothesis, namely, that the long-term average of an observable quantity coincides with its phase space average. The term {\em{mixing}} was introduced in Statistical Mechanics by Gibbs and its mathematical definition was introduced by von Neumann in 1932. See, for example~\cite{Gallavotti2016} and references therein. In the context of a dynamical system, there are two different notions of mixing - strong mixing and weak mixing - which are stronger than ergodicity. In particular, the following chain of implications hold: strong mixing $\implies$ weak mixing $\implies$ ergodicity (see Section~\ref{sec:review-ergodic}). Ergodicity and mixing are just two properties in the so-called {\em{ergodic hierarchy}} which is a central constituent of ergodic theory (see e.g.~\cite{ergodichierarchy}).

Quantum ergodicity was first discussed by von Neumann~\cite{vonNeumann1932} who established a quantum ergodic theorem. However, since von Neumann's work, there have been various different definitions of ergodicity and mixing in quantum systems (see e.g.~\cite{Peres84, Feingold84,Zhang2016}). These behaviours have also been studied extensively for quantum dynamical systems in the $C^*$-algebraic framework (see e.g.~\cite{Pillet2006, Fidaleo2009} and references therein).

The study of ergodicity and mixing in many-body quantum systems with local interactions is of particular relevance in the field of quantum chaos and has been the focus of much research (see e.g.~\cite{Bruno2019dual, Chan2018, Garratt2021} and references therein). The long-term behaviour of spatio-temporal correlation functions of local observables provide useful characterizations of ergodicity and mixing in such systems. It is often convenient to formulate the dynamics of such systems directly in terms of the time-evolution operator instead of the underlying Hamiltonian. Periodically kicked quantum systems, which are considered to be promising candidates of systems exhibiting quantum chaos, have been studied in this formulation, in particular, when the system is a one-dimensional lattice, i.e., a spin chain. The Floquet operator, which is the time evolution operator over one period, plays a key role in the dynamics of such systems. In fact, the latter is determined by the spectral properties of the Floquet operator (see e.g.~\cite{mccaw2005quantum}). In several studies involving spin chains~\cite{Bruno2019dual, Chan2018}, the time evolution over one period is comprised of two half steps. Each lattice site is coupled to its neighbour on one site via a bipartite unitary matrix in the first half step, and to its neighbour on the other side in the second half step. The time evolution of the spin chain is hence given by a quantum circuit which forms a brickwork pattern of unitary gates in $1+1$ dimensions. In~\cite{Chan2018}, the authors established that if the unitary matrices coupling adjacent lattice sites are random, drawn independently from the circular unitary ensemble (CUE), then the resulting system is a minimal model for many-body quantum chaos. In~\cite{Bruno2019dual}, the authors consider \emph{dual unitary} brickwork circuits (these are special unitary brickwork circuits in which the time evolution is unitary both in the temporal and spatial directions, see  Section~\ref{sec:lattice-models} for the precise definition) and proved that explicit computation of all dynamical correlations between pairs of local observables is possible in such circuits. Furthermore, it turns out that for any unitary brickwork circuit, the dual unitary constraint is necessary and sufficient to ensure that entanglement in the circuit spreads at the maximum rate possible \cite{zhou2022maximal}. For the case of a qubit spin chain, a classification of all dual unitary circuits based on varying degrees of ergodicity and mixing in the long-term behaviour of dynamical correlations has been obtained in \cite{Bruno2019dual}. A similar analysis was done for qudit spin chains (with local Hilbert space dimension $d > 2$) in \cite{Claeys2021dual}. 

A quantum channel, i.e., a linear completely positive trace-preserving map, provides a succinct representation of the dynamics of a general (possibly open) quantum system. Moreover, a quantum process can be characterized by a sequence of quantum channels. Thus, understanding the long-term behaviour of a channel (i.e., the behaviour of a channel under iterated compositions with itself) is of fundamental importance. The study of these ergodic properties for quantum channels in finite dimensions was done in~\cite{Burgarth2013ergodic} and for quantum processes in~ \cite{Movassagh2004}.

\subsection{Summary of main results}

In this paper, we consider quantum channels $\Phi$ acting on $d$-level quantum systems that are covariant with respect to diagonal orthogonal transformations: for any $d \times d$ complex matrix $X$,
\begin{equation}
    \forall O\in \mathcal{DO}_d: \qquad \Phi(OXO) = O\Phi(X)O,
\end{equation}
where $\mathcal{DO}_d$ denotes the group of all $d\times d$ diagonal orthogonal matrices.  These so-called {\em{diagonal orthogonal covariant}} (DOC) channels were introduced and thoroughly studied in \cite{Singh2021diagonalunitary}. Many important quantum channels studied in quantum information theory belong in this class (see \cite[Section 7]{Singh2021diagonalunitary}). Moreover, it can be shown that any DOC channel is parameterized by three $d\times d$ matrices $A,B$ and $C$, with $A$ being column stochastic and acting as the `classical core' of the channel (see Definition~\ref{def:DOC}). The central result of this paper establishes that the ergodic behaviour of any DOC channel is essentially governed by the underlying stochastic matrix $A$ (see Theorems~\ref{theorem:DOC-ergodic/mixing} and \ref{theorem:DOC-irred-prim}). Hence, for any such channel, ergodicity, mixing and the related properties of irreducibility and primitivity can all be easily verified by using results from the classical theory of ergodicity of stochastic matrices, or, more specifically, by analyzing the connectivity properties of a directed graph associated with $A$. The relation between the ergodic properties of a stochastic matrix and the connectivity properties of the associated directed graph is presented in some detail in Section~\ref{subsec:graphs-ergodic}.

As mentioned above, in~\cite{Bruno2019dual} it was shown that explicit computation of two-point spatio-temporal correlation functions of local observables is possible in dual unitary brickwork quantum circuits. If we additionally impose the following symmetry on the underlying unitary gates $U$:
\begin{equation}
    \forall O\in \mathcal{DO}_d: \qquad (O\otimes O)U(O\otimes O) = U,
\end{equation}
then the resulting subgroup of {\em{local diagonal orthogonal invariant}} (LDOI) unitary gates provides a rich and explicitly parameterized family of unitary gates which can be used to construct new families of highly symmetric dual unitary gates \cite{singh2022diagonal}. Moreover, the long term behaviour of two-point spatio-temporal correlations between local observables in a circuit built from such LDOI dual unitary gates is completely determined by the ergodic properties of a specific DOC quantum channel (Lemma~\ref{lemma:Lambda+LDOI}). This allows us to exploit our previously stated results on ergodic properties of DOC channels to show that brickwork circuits constructed from LDOI dual unitary gates form a rich class of circuits, in the sense that they exhibit all the possible kinds of (non-)ergodic behaviours listed (and analysed for qubit spin chains) in~\cite{Bruno2019dual} and for qudit spin chains in~\cite{Claeys2021dual}.

\subsection{Outline of the paper} In Section \ref{sec:review-ergodic} we review the basic definitions and results from classical and quantum ergodic theory, with an emphasis on the graph-theoretic point of view in the classical case (Subsection \ref{subsec:graphs-ergodic}). In Section \ref{sec:DOC-channels}, we discuss DOC channels and their ergodic properties, showing how they are related to the matrices defining them. Finally, Section \ref{sec:lattice-models} deals with brickwork quantum circuits constructed from LDOI dual unitary gates.

\section{A survey of classical and quantum ergodic theory}\label{sec:review-ergodic}

\subsection{Classical ergodic theory}\label{subsec:review-measure}
In a nutshell, classical ergodic theory can be described as the study of measure-preserving dynamics on measure spaces. Since all of the results that we present in this section are well-known in the literature, we simply state them without proofs. The interested readers should refer to the following references \cite{Viana2016ergodicbook, Cornfeld1982ergodicbook} for further details.

\begin{definition}
A dynamical system $(M,\mathfrak{S},\mu,T)$ is a measure space $(M,\mathfrak{S},\mu)$ equipped with a measurable transformation $T:M\to M$ that is measure-preserving:
\begin{equation*}
    \forall A\in \mathfrak{S}, \qquad \mu(T^{-1}A) = \mu(A),
\end{equation*}
where $T^{-1}A$ is the preimage of $A$ under $T$. We say that a set $A\in\mathfrak{S}$ is $T$-invariant if $T^{-1}A=A$. The map $T$ induces a linear isometry $U_T:L^2(\mu)\to L^2(\mu)$ defined as follows:
\begin{equation*}
    \forall f\in L^2(\mu), \qquad U_T(f) := f\circ T.
\end{equation*}
\end{definition}

We will always assume that our measure spaces $(M,\mathfrak{S},\mu)$ are normalized, i.e., $\mu(M)=1$. We now state some of the most important results in ergodic theory.

\begin{theorem}[Birkhoff-Kinchin] \label{theorem:Cergodic1}
Let $(M,\mathfrak{S},\mu,T)$ be a dynamical system and $f\in L^1(\mu)$. Then, 
\begin{equation*}
    \hat{f}(x) := \lim_{n\to \infty} \frac{1}{n} \sum_{k=0}^{n-1} f(T^k x) = \lim_{n\to \infty} \frac{1}{n} \sum_{k=0}^{n-1} U_T^k(f) (x) 
\end{equation*}
exists $\mu$ almost everywhere and $\hat{f}\in L^1(\mu)$. Moreover, $\hat{f}$ is $T$-invariant (i.e., $\hat f\circ T = \hat f$) and
\begin{equation*}
    \int \hat{f}(x)d\mu (x) = \int f(x) d\mu (x).
\end{equation*}
\end{theorem}

\begin{theorem}[Ergodicity] \label{theorem:Cergodic2}
For a dynamical system $(M,\mathfrak{S},\mu,T)$, the following are equivalent.
\begin{itemize}
    \item For all $T$-invariant sets $A\in\mathfrak{S}$, $\mu(A)$ is either $0$ or $1$.
    \item Any $T$-invariant function $f\in L^1(\mu)$ is constant almost everywhere.
    \item $\lambda=1$ is a simple eigenvalue of the linear isometry $U_T:L^2(\mu)\to L^2(\mu)$.
    \item For all functions $f\in L^1(\mu)$,
    \begin{equation*}
        \hat{f}(x) = \lim_{n\to \infty} \frac{1}{n} \sum_{k=0}^{n-1} f(T^k x) = \int f(x)d\mu (x) \quad \text{almost everywhere}.
    \end{equation*}
    \item For all measurable sets $A,B\in \mathfrak{S}$,
    \begin{equation*}
        \lim_{n\to \infty} \frac{1}{n}\sum_{k=0}^{n-1} \mu\left( T^{-k}A \cap B \right) = \mu(A)\mu(B).
    \end{equation*}
    \item For all square-integrable functions $f,g\in L^2(\mu)$,
    \begin{equation*}
        \lim_{n\to \infty} \frac{1}{n}\sum_{k=0}^{n-1} \int f(T^{k}x)g(x) d\mu(x) = \int f(x)d\mu(x) \int g(x)d\mu(x).
    \end{equation*}
\end{itemize}
A dynamical system $(M,\mathfrak{S},\mu,T)$ satisfying these equivalent conditions is said to be \emph{ergodic}.
\end{theorem}

\begin{theorem}[Weak mixing]
For a dynamical system $(M,\mathfrak{S},\mu,T)$, the following are equivalent.
\begin{itemize}
    \item For all measurable sets $A,B\in \mathfrak{S}$,
    \begin{equation*}
        \lim_{n\to \infty} \frac{1}{n}\sum_{k=0}^{n-1} \left\vert\mu\left( T^{-k}A \cap B \right) - \mu(A)\mu(B)\right\vert = 0.
    \end{equation*}
    \item For all square-integrable functions $f,g\in L^2(\mu)$,
    \begin{equation*}
        \lim_{n\to \infty} \frac{1}{n}\sum_{k=0}^{n-1}\left\vert \int f(T^{k}x)g(x) d\mu(x) - \int f(x)d\mu(x) \int g(x)d\mu(x) \right\vert = 0.
    \end{equation*}
\end{itemize}
A dynamical system $(M,\mathfrak{S},\mu,T)$ satisfying these equivalent conditions is said to be \emph{weakly mixing}.
\end{theorem}

\begin{theorem}[Strong mixing]
For a dynamical system $(M,\mathfrak{S},\mu,T)$, the following are equivalent.
\begin{itemize}
    \item For all measurable sets $A,B\in \mathfrak{S}$,
    \begin{equation*}
        \lim_{n\to \infty}  \mu\left( T^{-n}A \cap B \right) = \mu(A)\mu(B).
    \end{equation*}
    \item For all square-integrable functions $f,g\in L^2(\mu)$,
    \begin{equation*}
        \lim_{n\to \infty}  \int f(T^{n}x)g(x) d\mu(x) = \int f(x)d\mu(x) \int g(x)d\mu(x).
    \end{equation*}
\end{itemize}
A dynamical system $(M,\mathfrak{S},\mu,T)$ satisfying these equivalent conditions is said to be \emph{strongly mixing}.
\end{theorem}

These theorems allow us to define different levels of ergodicity of a classical dynamical system in a unified way by using the notion of correlation functions.

\begin{definition}\label{def:corr_measure}
Let $(M,\mathfrak{S},\mu,T)$ be a dynamical system. For any $f,g\in L^2(\mu)$, we define their correlation sequence as follows:
\begin{align*}
    C_{f,g}(n) &:= \int f(T^{n}x)g(x) d\mu(x) - \int f(x)d\mu(x) \int g(x)d\mu(x), \qquad n\in \mathbb{N}.
\end{align*}
\end{definition}

\begin{proposition}
A dynamical system $(M,\mathfrak{S},\mu,T)$ is
\begin{itemize}
    \item \emph{ergodic} if and only if $\lim_{n\to \infty} \frac{1}{n}\sum_{k=0}^{n-1} C_{f,g}(k) = 0$ for all $f,g\in L^2(\mu)$,
    \item \emph{weakly mixing} if and only if $\lim_{n\to \infty} \frac{1}{n}\sum_{k=0}^{n-1} |C_{f,g}(k)| = 0$ for all $f,g\in L^2(\mu)$,
    \item \emph{strongly mixing} if and only if $\lim_{n\to \infty} C_{f,g}(n) = 0$ for all $f,g\in L^2(\mu)$.
\end{itemize}
\end{proposition}

It can be easily shown that a strongly mixing dynamical system is also weakly mixing, which further implies that the system is ergodic:

\begin{equation}
    \text{strong mixing} \implies \text{weak mixing} \implies \text{ergodicity}.
\end{equation}

\subsection{Ergodic theory of quantum channels and stochastic matrices} \label{subsec:Qergodic}
We denote the algebra of all $d\times d$ complex matrices by $\M{d}$. For $A\in \M{d}$, the operations of transposition, entrywise complex conjugation, and conjugate transposition, respectively, are denoted by $A^\top, \bar A,$ and $A^\dagger$. We denote the convex set of \emph{quantum states} (positive semi-definite matrices with unit trace) in $\M{d}$ by $\St{d}$. A linear map $\Lambda:\M{d}\to \M{d}$ is called a \emph{quantum channel} if it is completely positive and trace-preserving. The set of eigenvalues (i.e., the spectrum) $\operatorname{spec}\Lambda$ of a quantum channel $\Lambda$ is contained within the unit disk $\{ z\in \mathbb{C}: |z|\leq 1\}$ in the complex plane and is invariant under complex conjugation, i.e., $\lambda\in \operatorname{spec}\Lambda \implies \overbar \lambda \in \operatorname{spec}\Lambda$. The peripheral spectrum of $\Lambda$ consists of all peripheral eigenvalues $\lambda \in \mathbb{T} \cap \operatorname{spec}\Lambda$, where $\mathbb{T}:= \{z\in \mathbb{C} : |z|=1 \}$. It is known that the geometric and algebraic multiplicities of all peripheral eigenvalues of a quantum channel are equal \cite[Proposition 6.2]{Wolf2012Qtour}. We say that a peripheral eigenvalue is simple if it has unit multiplicity. For a more elaborate discussion on the properties and spectra of quantum channels, we refer the readers to \cite{Wolf2012Qtour}. The following result should be considered as the quantum analogue of Theorem~\ref{theorem:Cergodic1}.

\begin{theorem}\label{theorem:Qergodic1}
Let $\Lambda:\M{d}\to \M{d}$ be a quantum channel. Then, the limit
\begin{equation*}
    \hat \Lambda := \lim_{n\to \infty} \frac{1}{n} \sum_{k=0}^{n-1} \Lambda^n
\end{equation*}
defines another quantum channel $\hat{\Lambda}:\M{d}\to \M{d}$ satisfying $\Lambda\circ \hat\Lambda = \hat\Lambda$.
\end{theorem}
\begin{proof}
See \cite[Proposition 6.3]{Wolf2012Qtour}.
\end{proof}

We now present several equivalent ways to define the notion of quantum ergodicity in complete correspondence with the classical theory (see Theorem~\ref{theorem:Cergodic2}).

\begin{theorem}\label{theorem:Qergodic2}
For a quantum channel $\Lambda:\M{d}\to \M{d}$, the following are equivalent.
\begin{itemize}
    \item $\lambda=1$ is a simple eigenvalue of $\Lambda$.
    \item There exists a state $\rho\in \St{d}$ such that for all $A\in \M{d}$,
    \begin{equation*}
        \hat \Lambda (A) = \lim_{n\to \infty} \frac{1}{n} \sum_{k=0}^{n-1} \Lambda^k(A) = \operatorname{Tr}(A)\rho.
    \end{equation*}
    \item There exists a state $\rho\in\St{d}$ such that for all $A,B\in \M{d}$,
    \begin{equation*}
        \lim_{n\to \infty} \frac{1}{n} \sum_{k=0}^{n-1} \operatorname{Tr}\left(\Lambda^k(A) B\right) = \operatorname{Tr}(A)\operatorname{Tr}(B\rho).
    \end{equation*}
\end{itemize}
A quantum channel $\Lambda:\M{d}\to \M{d}$ satisfying these equivalent conditions is said to be \emph{ergodic}. The unique fixed point $\rho\in \St{d}$ of $\Lambda$ is called the \emph{stationary state} of $\Lambda$.
\end{theorem}
\begin{proof}
Let $\lambda=1$ be a simple eigenvalue of $\Lambda$ and $\rho\in\St{d}$ be the associated eigenstate. Then, since $\hat \Lambda$ is trace-preserving and its range is contained in $\operatorname{span}\{\rho\}$ (see  Theorem~\ref{theorem:Qergodic1}), it is clear that $\hat\Lambda(A)=\operatorname{Tr}(A)\rho$ for all $A\in \M{d}$. Conversely, if $\hat\Lambda(A)=\operatorname{Tr}(A)\rho$ for all $A\in \M{d}$, then Theorem~\ref{theorem:Qergodic1} again informs us that $\Lambda(\rho)=\rho$. Moreover, for any $\sigma\in\M{d}$ satisfying $\Lambda(\sigma)=\sigma$,  
\begin{equation*}
    \hat\Lambda(\sigma)=\sigma=\operatorname{Tr}(\sigma)\rho\implies \rho = \sigma.
\end{equation*}
The equivalence of the last two statements is straightforward to establish. 
\end{proof}

\begin{remark}\label{remark:corr_Q}
Let $\Lambda:\M{d}\to \M{d}$ be a quantum channel. For any $A,B\in \M{d}$, we define their \emph{correlation sequence} with respect to $\rho\in\M{d}$ as follows
\begin{equation*}
    C^\rho_{A,B}(n) := \operatorname{Tr}\left( \Lambda^n(A)B\right ) - \operatorname{Tr}(A)\operatorname{Tr}(B\rho), \qquad n\in\mathbb N.
\end{equation*}
Clearly, $\Lambda$ is ergodic if and only if there exists a state $\rho\in \St{d}$ such that for all $A,B\in \M{d}$, 
\begin{equation*}
\lim_{n\to \infty} \frac{1}{n}\sum_{k=0}^{n-1} C^{\rho}_{A,B}(k) = 0.    
\end{equation*}
\end{remark}

The following theorem shows that the classically distinct notions of weak and strong mixing are equivalent for quantum channels.

\begin{theorem}
For a quantum channel $\Lambda:\M{d}\to \M{d}$, the following are equivalent.
\begin{enumerate}
    \item (Weak mixing) There exists a state $\rho\in\St{d}$ such that for all $A,B\in \M{d}$, 
    \begin{equation*}
    \lim_{n\to \infty} \frac{1}{n}\sum_{k=0}^{n-1} |C^{\rho}_{A,B}(k)| = 0.    
\end{equation*}
     \item (Strong mixing) There exists a state $\rho\in\St{d}$ such that for all $A,B\in \M{d}$,
    \begin{equation*}
    \lim_{n\to \infty} C^{\rho}_{A,B}(n) = 0.    
\end{equation*}
    \item There exists a state $\rho\in\St{d}$ such that for all $A\in \M{d}$, 
    \begin{equation*}
    \lim_{n\to \infty} \Lambda^n(A) = \operatorname{Tr}(A)\rho.    
    \end{equation*}
    \item $\lambda=1$ is a simple eigenvalue of $\Lambda$ and there are no other peripheral eigenvalues of $\Lambda$.
\end{enumerate}
A quantum channel satisfying these equivalent conditions is said to be \emph{mixing}.
\end{theorem}
\begin{proof}
$(4)\implies (3)\implies (2)$ Assume that $\Lambda$ has no peripheral eigenvalues except $\lambda=1$ which is simple and let $\rho$ be the associated eigenstate. Then, we can split $\M{d}$ as \cite[Chapter 8]{Axler2014linear}:
\begin{equation*}
    \M{d} = \operatorname{span}\{\rho\} \oplus \bigoplus_{\lambda\in \operatorname{spec}\Lambda, |\lambda|<1} G_{\lambda}(\Lambda),
\end{equation*}
where $G_{\lambda}(\Lambda)$ is the generalized eigenspace associated with the eigenvalue $\lambda$ and $\Lambda\vert_{G_\lambda (\Lambda) } = \lambda \operatorname{id} + N_\lambda$ for all $|\lambda|<1$, where $\operatorname{id}:G_{\lambda}(\Lambda)\to G_{\lambda}(\Lambda)$ is the identity map and $N_\lambda:G_{\lambda}(\Lambda)\to G_{\lambda}(\Lambda)$ is nilpotent. Note that since peripheral eigenvalues have equal algebraic and geometric multiplicities, $G_1(\Lambda)=\operatorname{span}\{\rho\}$. Clearly, $\Lambda^n\vert_{G_\lambda (\Lambda) } \to 0$ as $n\to \infty$ if $|\lambda|<1$. Hence, for all $A\in\M{d}$, we have
\begin{equation*}
    \lim_{n\to \infty} \Lambda^n (A) = \operatorname{Tr}(A)\rho,
\end{equation*}
which clearly implies (2). \\
$(2)\implies (1)$ This is easy to show by using the convergence of Ces\`aro means, i.e., $a_n\to a \implies s_n\to a$, where $s_n := \frac{1}{n} \sum_{i=1}^n a_i$. \\
$(1)\implies (4)$ Condition $(1)$ implies that $\Lambda:\M{d}\to \M{d}$ is ergodic and hence $\lambda=1$ is a simple eigenvalue of $\Lambda$. Now, assume that $\Lambda$ has another peripheral eigenvalue $\gamma\neq 1$. Let $A\in\M{d}$ with $\operatorname{Tr}(A)=0$ be the associated eigenmatrix. Note that $A$ can be chosen to be traceless since $\mathbb{1}_d$ is a left eigenmatrix of $\Lambda$, i.e., $\Lambda^{\dagger}(\mathbb{1}_d)=\mathbb{1}_d$. Then, $A$ must be the zero matrix because
\begin{align*}
    \lim_{n\to \infty} \frac{1}{n}\sum_{k=0}^{n-1} |C^{\rho}_{A,A^\dagger}(k)| = \lim_{n\to \infty} \frac{1}{n}\sum_{k=0}^{n-1} |\operatorname{Tr}(\Lambda^k(A)A^\dagger)| &= |\operatorname{Tr}(AA^\dagger)| \lim_{n\to \infty} \frac{1}{n}\sum_{k=0}^{n-1} |\gamma^k|  \\
    &= |\operatorname{Tr}(AA^\dagger)|=0.
\end{align*}
Thus, $\Lambda$ has no peripheral eigenvalues except $\lambda=1$.
\end{proof}

\begin{definition}
A quantum channel is said to be 
\begin{itemize}
    \item \emph{irreducible} if it is ergodic and if its stationary state is positive definite.
    \item \emph{primitive} if it is mixing and if its stationary state is positive definite.
\end{itemize}
\end{definition}

\begin{remark}
Several other equivalent characterizations of the classes of irreducible and primitive quantum channels can be found in \cite[Chapter 6]{Wolf2012Qtour}. For an elaborate study of ergodic and mixing quantum channels, we refer the readers to \cite{Burgarth2013ergodic} and references therein.
\end{remark}

\begin{remark}\label{remark:ergodic-unital}
A quantum channel $\Lambda:\M{d}\to \M{d}$ is said to be \emph{unital} if $\iden_d$ is fixed by $\Lambda$, i.e. $\Lambda(\iden_d)=\iden_d$. Clearly, the ergodic and mixing properties for unital channels are equivalent to irreducibility and primitivity, respectively.
\end{remark}

One can recover the ergodic theory of stochastic matrices from the quantum ergodic theory that we have introduced above. Recall that a matrix $A\in \M{d}$ is said to be column stochastic if it is entrywise non-negative and its column-sums are one: for all $j\in [d]$, $\sum_{i=1}^d A_{ij} =1$. We should emphasize here that if $A$ is considered to be the transition matrix of a time-homogeneous Markov chain with finite state space, then the following discussion can be formulated entirely in the language of Markov chains.

For $A\in \M{d}$, consider the linear map $\Phi_A:\M{d}\to \M{d}$ defined as 
\begin{equation}
    \forall X\in \M{d}: \quad \Phi_A (X) = \sum_{i,j}  A_{ij}X_{jj} \ketbra{i}.
\end{equation}
It is easy to show that $\Phi_A$ is completely positive and trace-preserving if and only if $A$ is column stochastic (see Theorem~\ref{theorem:DOC-cptp}). Moreover, $\operatorname{spec}\Phi_A = \operatorname{spec}A \cup \{0\}$ (see Theorem~\ref{theorem:DOC-spec}). Hence, all the properties of the spectrum of a quantum channel that were introduced in the beginning of this section also hold for the spectrum of a stochastic matrix. We can thus derive the ergodic properties of $A$ from those of $\Phi_A$. This is done in the following two theorems, the proofs of which are left for the reader. Note that we denote by $\ket{e}$ the vector in $\mathbb{R}^d$ with each component equal to $1$. The standard probability simplex in $\mathbb{R}^d$ is denoted by $\Delta^d$.

\begin{theorem}\label{theorem:stoch-ergodic}
For a column stochastic matrix $A\in \M{d}$, the following are equivalent.
\begin{itemize}
    \item $\Phi_A$ is ergodic.
    \item $\lambda=1$ is a simple eigenvalue of $A$. 
    \item There exists a probability distribution $\pi\in\Delta^d$ such that
    \begin{equation*}
        \hat A = \lim_{n\to \infty} \frac{1}{n}\sum_{k=0}^{n-1} A^k = \ketbra{\pi}{e}.
    \end{equation*}
    \item There exists a probability distribution $\pi\in\Delta^d$ such that for all $v,w\in \C{d}$,
    \begin{equation*}
        \lim_{n\to \infty} \frac{1}{n}\sum_{k=0}^{n-1} \langle v|A^k|w\rangle = \langle v\ketbra{\pi}{e} w\rangle.
    \end{equation*}
\end{itemize}
A column stochastic matrix $A$ satisfying these equivalent conditions is said to be \emph{ergodic}. The unique fixed distribution $\pi\in\Delta^d$ of $A$ is said to be the \emph{stationary distribution} of $A$.
\end{theorem}

\begin{theorem}\label{theorem:stoch-mixing}
For a column stochastic matrix $A\in \M{d}$, the following are equivalent.
\begin{itemize}
    \item $\Phi_A$ is mixing.
    \item $\lambda=1$ is a simple eigenvalue of $A$ and $A$ has no other peripheral eigenvalues.
    \item There exists a probability distribution $\pi\in\Delta^d$ such that
    \begin{equation*}
        \lim_{n\to \infty} A^n = \ketbra{\pi}{e}.
    \end{equation*}
    \item There exists a probability distribution $\pi\in\Delta^d$ such that for all $v,w\in \C{d}$,
    \begin{equation*}
        \lim_{n\to \infty} \langle v|A^n|w\rangle = \langle v\ketbra{\pi}{e} w\rangle.
    \end{equation*}
\end{itemize}
A column stochastic matrix $A$ satisfying these equivalent conditions is said to be \emph{mixing}.
\end{theorem}

\begin{remark}
Mixing stochastic matrices as defined by the above theorem are known by at least a couple of different names in the literature. In \cite{Seneta1981nonnegative}, they are called \emph{regular} stochastic matrices. In \cite{Wolfowitz1963SIA}, they are referred to as SIA (stochastic, indecomposable, aperiodic) matrices.
\end{remark}

\begin{definition}
A column stochastic matrix $A\in\M{d}$ is said to be
\begin{itemize}
    \item \emph{irreducible} if it is ergodic and its stationary distribution has full support.
    \item \emph{primitive} if it is mixing and its stationary distribution has full support. 
\end{itemize}
\end{definition}

It should be emphasized that the stated spectral chracterizations are not really useful in practice to determine if a given quantum channel or a stochastic matrix has a desired ergodic propertry. This is because the Abel–Ruffini theorem guarantees that the eigenvalues of a $d\times d$ matrix cannot in general be computed by radicals if $d>4$. Development of verifiable characterizations of quantum ergodicity has been an important theme of research in the past decade \cite{Burgarth2013ergodic, Jamiokowski2014shemesh}. For classical stochastic matrices, however, efficient descriptions of ergodicity are much easier to develop than for quantum channels. Indeed, it can be shown that the pattern of distribution of non-zero entries within a stochastic matrix completely determines its ergodic behaviour. In other words, the study of ergodicity of a stochastic matrix is precisely the study of certain connectivity properties of its directed graph. We make this connection more precise in the next section.

\subsection{Graph-theoretic description of ergodicity of stochastic matrices} \label{subsec:graphs-ergodic} Before establishing the aforementioned connection, it is essential to introduce some basic graph-theoretic terminology.

\begin{definition}
A $d$-vertex \emph{directed graph} (or digraph) $G=(V,E)$ consists of a vertex set $V=[d]:=\{1,2,\ldots ,d\}$ and a collection $E$ of ordered pairs of vertices called edges.
\end{definition}

We allow our graphs to have loops but not multiple edges. A digraph $H=(V_H,E_H)$ is a subgraph of $G=(V,E)$ (denoted $H\subset G$) if $V_H\subset V$ and $E_H\subset E$. Given a digraph $G=(V,E)$, every subset of vertices $W\subset V$ defines an \emph{induced subgraph} $G(W)=(W,E_W) \subset G$, where $E_W$ is the set of all edges that start and end in $W$. For two distinct vertices $i,j\in V$, we say that
\begin{itemize}
    \item $i$ \emph{leads} to $j$ in $n$ steps (denoted $i\xrightarrow{n} j$) if there exist vertices $i_1,\ldots ,i_{n-1}\in V$ such that 
    \begin{equation*}
         (i,i_1),(i_1,i_2),\ldots ,(i_{n-1},j)\in E.
    \end{equation*}
    \item $i$ \emph{leads} to $j$ (denoted $i\to j$) if there exists some $n\in\mathbb{N}$ such that $i\xrightarrow{n} j$.
    \item $i$ \emph{communicates} with $j$ (denoted $i\leftrightarrow j$) if $i\to j$ and $j\to i$.
\end{itemize}

Similarly, for two subgraphs $H=(V_H,E_H), H'=(V_{H'},E_{H'}) \subset G$, we say that 
\begin{itemize}
    \item $H$ leads to $H'$ (denoted $H\to H'$) if there exist vertices $h\in V_H, h'\in V_{H'}$ such that $h\to h'$.
    \item $H$ communicates with $H'$ (denoted $H\leftrightarrow H'$) if $H\to H'$ and $H'\to H$.
\end{itemize}

A subgraph $H=(V_H,E_H)\subset G$ is said to be \emph{closed} if $i\in V_H$ and $i\to j \implies j\in V_H$. In other words, there is no escape once you land in a closed subgraph. $H$ is said to be \emph{fully accessible} if every vertex not in $V_H$ leads to some vertex in $V_H$, i.e.,~$\forall i\in V\setminus V_H, \exists j\in V_H$ such that $i\to j$.

\begin{definition}
A $d$-vertex digraph ($d\geq 2$) $G=(V,E)$ is said to be
\begin{itemize}
    \item \emph{strongly connected} if all pairs of distinct vertices $i,j\in V$ communicate with each other.
    \item \emph{aperiodic} if there exists $n\in \mathbb{N}$ such that for all distinct vertices $i,j\in V$, $i\xrightarrow{n} j$.
\end{itemize} 
\end{definition}

\begin{remark}
A digraph $G$ is said to contain a \emph{cycle} of length $n$ if there exists a vertex $i$ such that $i$ leads to itself in $n$ steps, i.e.,~$i\xrightarrow{n} i$. The \emph{period} of a vertex $i$ of a digraph $G$ is defined to be the greatest common divisor of the lengths of all cycles originating from $i$. If $G$ is strongly connected, then all its vertices have the same period, which is defined to be the period of $G$ itself. In this terminology, a strongly connected digraph $G$ is said to \emph{aperiodic} if it has unit period. It is easy to check that the two definitions of aperiodicity presented above are equivalent. 
\end{remark}

We define single vertex graphs to be strongly connected (or aperiodic) if they contain a loop. Note that an aperiodic digraph is always strongly connected. If we assume that all vertices $i\in V$ of a digraph $G=(V,E)$ are self-communicating, i.e., $i\leftrightarrow i$, then $\leftrightarrow$ defines an equivalence relation on $V$, thus partitioning it into disjoint \emph{communicating classes}: $V=\cup_{k=1}^n C_k$. Every class $C_k$ with two or more vertices gives rise to a strongly connected subgraph $G(C_k)\subset G$. Note that if a class consists of a single vertex $i$, it does not necessarily imply the existence of the loop $(i,i)\in E$.

\begin{remark}\label{remark:closed-class}
We can define a class $C$ of a graph to be closed, fully accessible, or aperiodic depending on whether its induced subgraph $G(C)$ is closed, fully accessible, or aperiodic. In \cite{Snell1983markov}, a closed class is called \emph{ergodic}, while a non-closed class is called \emph{transient}. 
\end{remark}

\begin{figure}[H]
    \centering
    \includegraphics[scale=1.4]{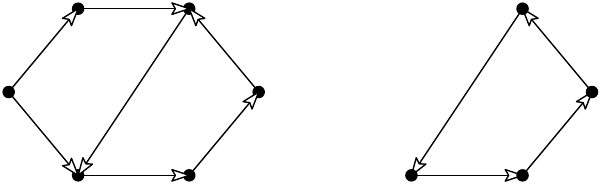}
    \caption{A 6-vertex graph $G$ and a 4-vertex induced subgraph $H$. It can be easily checked that $H$ is strongly connected with period $=4$. Moreover, $H$ is closed and fully accessible. }
\end{figure}

Given $A\in \M{d}$, we define its associated digraph $G_A=(V,E)$ with the vertex set $V=[d]$ and $(i,j)\in E \iff A_{ji}\neq 0$ for all $i,j\in V$. If $A$ is column stochastic, we can think of it as the transition matrix of a time homogeneous Markov chain defined on the state space $V=[d]$, with the edge $(i,j)$ representing a non-zero probability $A_{ji}$ of transitioning from $i$ to $j$. Note that our convention is different from the standard one of choosing transition matrices to be row stochastic. 

From the definition of strong connectivity of a graph, the following lemma immediately follows.

\begin{lemma}\label{lemma:reducible}
For a matrix $A\in \M{d}$, the following are equivalent.
\begin{itemize}
    \item $G_A$ is not strongly connected.
    \item There exists a permutation matrix $P\in \M{d}$ such that 
    \begin{equation*}
        PAP^\top = \left(
\begin{array}{ c c }
   B_{1,1} & B_{1,2} \\
   0 & B_{2,2}
\end{array}
\right)
    \end{equation*}
    for some square matrices $B_{11},B_{22}$.
\end{itemize}
If $A$ satisfies these equivalent conditions, it is said to be \emph{reducible}.
\end{lemma}

Note that we have deliberately introduced the term reducible in such a way that for a column stochastic matrix, the property of not being reducible becomes equivalent to its ergodic property of irreducibility, as is illustrated by the following result.

\begin{theorem}\label{theorem:perron}
A column stochastic matrix $A\in \M{d}$ is
\begin{itemize}
    \item irreducible if and only if $G_A$ is strongly connected.
    \item primitive if and only if $G_A$ is aperiodic.
\end{itemize}
\end{theorem}
\begin{proof}
These equivalences are some of the core results of the Perron-Frobenius theory of non-negative matrices (see \cite[Chapter 1]{Seneta1981nonnegative}).
\end{proof}

A repeated application of Lemma~\ref{lemma:reducible} implies that every reducible matrix $A\in \M{d}$ can be brought into a canonical upper block-triangular form via a permutation conjugation:
\begin{equation}\label{eq:canonical}
    PAP^\top = \left(   \begin{array}{cccc}
         A_{1,1} & A_{1,2} & \ldots & A_{1,n}  \\
         0 & A_{2,2} & & A_{2,n} \\
         \vdots & & \ddots &  \\
         0 & 0 & \ldots & A_{n,n}
    \end{array}  \right),
\end{equation}
where each square matrix $A_{k,k}$ is either not reducible (= irreducible) or is the $1\times 1$ zero matrix. Note that the vertices in each block $A_{k,k}$ form one of the communication classes $C_k$ of the vertex set of $G_A$, which have been arranged by the permutation $P$ in such a way that $G_A(C_l)$ is accessible from $G_A(C_k)$ only when $l< k$. More precisely, for $l\neq k$, $G_A(C_k)\to G_A(C_l)$ if and only if $A_{l,k}$ is a non-zero block. Hence, the class $C_1$ is clearly closed. Moreover, the eigenvalues of $A$ are simply the eigenvalues of the diagonal blocks:
\begin{equation}
    \operatorname{spec} A = \cup_{k=1}^n \operatorname{spec} A_{k,k}.
\end{equation}

\begin{figure}[htp!]
    \centering
    \includegraphics[scale=1.2]{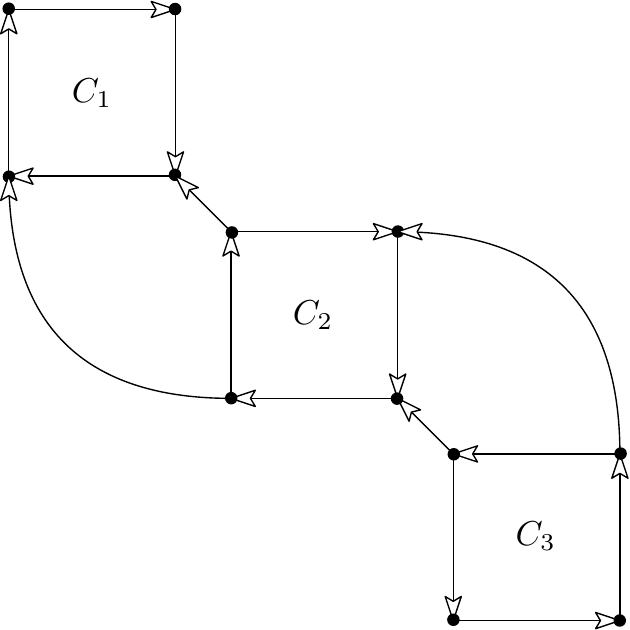}
    \caption{A $12$-vertex graph $G_A$ depicting the class structure of some $12\times 12$ reducible matrix $A$. The classes have been drawn to mimic the canonical form (Eq.~\eqref{eq:canonical}) of a reducible matrix. Each class $C_k$ corresponds to a $4\times 4$ irreducible block $A_{k,k}$. Connections between classes $G(C_k)\to G(C_l)$ depict that the $A_{l,k}$ block is non-zero for $l<k$. In this case, $C_1$ is the only class that is closed. Hence, if $A$ is column stochastic, then in accordance with Corollary~\ref{corollary:Cergodic}, it must be ergodic. The classes $C_2$ and $C_3$ are transient (see Remark \ref{remark:closed-class}).}
    \label{fig:ergodic_graph}
\end{figure}

Now, if $A$ is stochastic, then $A_{1,1}$ is stochastic as well (and thus can't be the $1\times 1$ zero matrix). Hence, the closed class $C_1$ gives rise to a closed and strongly connected subgraph $G_A(C_1)\subset G_A$. It turns out that ergodicity of $A$ is equivalent to demanding full accessibility of this subgraph, which is equivalent to demanding that none of the other subgraphs $G_A(C_k)$ (for $k\geq 2$) are closed. In other words, $A$ is ergodic if and only if $G_A$ has a unique closed class. More generally, it is possible to show that $\lambda=1$ is an eigenvalue of $A$ with multiplicity $n$ if and only if $G_A$ has precisely $n$ closed classes. With the requisite concepts in hand, we can now state and prove this result.

\begin{theorem}
The multiplicity of the $\lambda=1$ eigenvalue of a column stochastic matrix $A\in \M{d}$ is equal to the number of closed classes in $G_A$.
\end{theorem}
\begin{proof}
If $A$ is not reducible (= irreducible), then the desired result follows from Theorem~\ref{theorem:perron}. If $A$ is reducible, we bring it to its canonical form given below via a permutation $P$ (see Eq.~\eqref{eq:canonical}).
Recall that each diagonal block $A_{k,k}$ is either not reducible or is the $1\times 1$ zero matrix and corresponds to one of the communication classes $C_k$ of $G_A$. We now show that the following are equivalent.
\begin{enumerate}
    \item The class $C_k$ is closed.
    \item $A_{k,k}$ is column stochastic.
    \item $\lambda=1$ is an eigenvalue of $A_{k,k}$.
\end{enumerate}

\textit{Proof of equivalence} $(1)\implies (2)$: Since the class $C_k$ is closed, all blocks $A_{l,k}=0$ for $l<k$. Since $A$ is column stochastic, this implies that $A_{k,k}$ is column stochastic. $(2)\implies (3)$: Since $A_{k,k}$ is column stochastic and irreducible, the conclusion follows from Theorem~\ref{theorem:stoch-ergodic}. $(3)\implies (1)$: This is the crucial implication. Assume on the contrary that $C_k$ is not closed. Then, $A_{l,k}\neq 0$ for some $l<k$ (i.e. $G(C_k)\to G(C_l)$ for some $l<k$). If $A_{l,l}$ is the $1\times 1$ zero block, then $A_{l',l}\neq 0$ for some $l'<l$ (i.e. $G(C_l)\to G(C_{l'})$ for some $l'<l$). Continuing in this way, we must eventually land in an irreducible block. Hence, we assume that $A_{l,l}$ itself is irreducible (i.e. $G(C_l)$ is strongly connected). Thus, every vertex $i\in C_k$ leads to some vertex in $C_l$ in say, $q_i$ steps. Once we land in $C_l$, we can use the strong connectivity of $G_A(C_l)$ to show that every vertex $i\in C_k$ leads to some vertex in $C_l$ also in $q_i+p$ steps for any $p\in \mathbb{N}$. Hence, if $q:=\max_{i\in C_k} q_i$, every vertex $i\in C_k$ leads to some vertex in $C_l$ in $q$ steps. This means that in the stochastic matrix
\begin{equation}
    PA^qP^\top = \left(   \begin{array}{cccc}
         A^q_{1,1} & * & \ldots & * \\
         0 & A^q_{2,2} & & * \\
         \vdots & & \ddots & \vdots  \\
         0 & 0 & \ldots & A^q_{n,n}
    \end{array}  \right),
\end{equation}
every column $i\in C_k$ has a non-zero entry. Hence, every column sum of the block $A^q_{k,k}$ is strictly less than one. Thus, we can employ the Gershgorin circle theorem \cite[Chapter 6]{Horn2012matrix} to conclude that 
\begin{equation*}
  \forall \lambda\in \operatorname{spec}A_{k,k}, \quad  |\lambda|^q<1 \implies |\lambda|<1.
\end{equation*}
This contradicts the fact that $\lambda=1$ is an eigenvalue of $A_{k,k}$. Hence, $C_k$ must be closed.

With the crucial equivalences established, the conclusion of the theorem follows by noting that
\begin{equation}
    \operatorname{spec} A = \cup_{k=1}^n \operatorname{spec} A_{k,k},
\end{equation}
and if $\lambda=1$ is an eigenvalue of any of the blocks $A_{k,k}$, then it must be simple because all the blocks are known to be not reducible.
\end{proof}

The two corollaries stated below follow trivially from the above result.

\begin{corollary}\label{corollary:Cergodic}
For a column stochastic matrix $A\in \M{d}$, the following are equivalent. 
\begin{itemize}
    \item $A$ is ergodic.
    \item $G_A$ has a unique closed class.
    \item There exists a closed, fully accessible, and strongly connected subgraph $G\subset G_A$.
\end{itemize}
\end{corollary}

\begin{corollary}\label{corollary:Cmixing}
For a column stochastic matrix $A\in \M{d}$, the following are equivalent. 
\begin{itemize}
    \item $A$ is mixing.
    \item $G_A$ has a unique closed and aperiodic class.
    \item There exists a closed, fully accessible, and aperiodic subgraph $G\subset G_A$.
\end{itemize}
\end{corollary}

The results of the above corollaries are present in at least two different forms in the literature, which we now discuss. Let us define a matrix $A\in \M{d}$ to be \emph{decomposable} (or \emph{completely reducible}) if there exists a permutation matrix $P\in \M{d}$ such that
\begin{equation}
    PAP^\top =  \left(   \begin{array}{cc}
         B_{1,1} & 0  \\
         0 & B_{2,2} 
    \end{array}  \right),
\end{equation}
where $B_{1,1}, B_{2,2}$ are square matrices. It can easily be verified that a stochastic matrix $A$ is not decomposable if and only if there exists a closed, fully accessible, and strongly connected subgraph $G\subset G_A$, which is equivalent to ergodicity of $A$ according to Corollary~\ref{corollary:Cergodic}. Imposing aperiodicty on $G_A$ would then give us the mixing property. This is why mixing stochastic matrices are also known as SIA (stochastic, indecomposable, aperiodic) matrices in the literature \cite{Wolfowitz1963SIA}. 

Another characterization of mixing stochastic matrices comes from the notion of scrambling \cite{Dobrushin1956scrambling, Hajnal1958weakergo, Seneta1979coeffergo,Akelbek2009scrindex}. A digraph $G=(V,E)$ is said to be scrambling if for any two vertices $i,j\in V$, there exists $k\in V$ such that both $i$ and $j$ lead to $k$ in one step, i.e., $(i,k), (j,k)\in E$. The \emph{scrambling index} of $G$ (denoted $n(G)$) is the smallest positive integer $n$ such that for any two vertices $i,j\in V$, there exists $k\in V$ such that both $i$ and $j$ lead to $k$ in $n$ steps: $i\xrightarrow{n} k$, $j\xrightarrow{n} k$. If no such $n$ exists, we say that $n(G)=0$. We can define a column stochastic matrix $A\in \M{d}$ to be \emph{scrambling} if $G_A$ is scrambling, i.e., $A$ is scrambling if for any two columns $i,j\in [d]$, there exists a row $k\in [d]$ such that $A_{ki}A_{kj}>0$. Then, $n(G_A)$ is the minimum positive integer $n$ such that $A^n$ is scrambling. It has recently been shown that for a digraph $G$, $n(G)> 0$ if and only if it has a closed, fully accessible, and aperiodic subgraph, \cite[Theorem 4.1]{Guterman2019scrindex}. We can recover this result from Corollary~\ref{corollary:Cmixing} with the help of the following simple lemma.

\begin{lemma}
A column stochastic matrix $A\in \M{d}$ is mixing if and only if $n(G_A)>0$.
\end{lemma}
\begin{proof}
It is known that if $A$ is scrambling, then all its non-unit eigenvalues $\lambda$ satisfy $|\lambda|<1$ \cite{Seneta1979coeffergo}. Thus, if $n(G_A)>0$, then all the non-unit eigenvalues of $A$ must be non-peripheral, i.e., $A$ is mixing. Conversely, if $A$ is mixing, then all columns of $A^n$ converge to the same limit as $n\to \infty$ (see Theorem~\ref{theorem:stoch-mixing}). Hence, for a large enough $n$, $A^n$ must necessarily be scrambling.
\end{proof}

\section{Ergodicity of DOC channels}\label{sec:DOC-channels}

In the previous section, we saw how the ergodic theory of stochastic matrices can be derived from quantum ergodic theory by restricting to a special family of `classical' channels:
\begin{align}
    \Phi_A:\M{d} &\to \M{d} \nonumber \\ 
    X &\mapsto \operatorname{diag}(A\ket{\operatorname{diag}X}) = \sum_{i,j}  A_{ij}X_{jj} \ketbra{i}.
\end{align}

We say that these channels are classical because for any input state $\rho\in\St{d}$, $\Phi_A$ acts only by applying the stochastic matrix $A\in\M{d}$ to the diagonal probability distribution $\ket{\operatorname{diag}\rho}\in \Delta^d$, i.e., the vector of diagonal elements of $\rho$. All the ergodic properties of $\Phi_A$ are equivalent to those of $A$, and thus can be efficiently verified by Theorem~\ref{theorem:perron} and Corollary~\ref{corollary:Cergodic}. A formal proof of this statement will follow from a more general result that we will prove in this section. In a nutshell, our idea is to extend this family of classical channels in such a way that even though the resulting channels are genuinely quantum, all the ergodic properties of these channels would still stem from the underlying `classical core'. With this end in sight, let us start by defining the said class of quantum channels, which were introduced and thoroughly examined in \cite{Singh2021diagonalunitary}. In order to make the exposition more visually intuitive, we will often make use of diagrams comprised of boxes and wires to depict tensors. For those who are unfamiliar with this language, \cite[Section 3]{Nechita2021} should suffice for a quick introduction. 

\begin{definition}\label{def:DOC}
For $A,B,C\in \M{d}$ with $\operatorname{diag}A=\operatorname{diag}B=\operatorname{diag}C$, we define
\begin{alignat*}{2}
    \Phi^{(1)}_{(A,C)}:\M{d} &\to \M{d} \\ 
    X &\mapsto \includegraphics[scale=1, align=c]{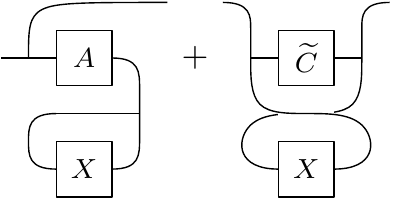} = \operatorname{diag}(A\ket{\operatorname{diag}X}) + \widetilde{C}\,\odot X^\top  \\
    \Phi^{(2)}_{(A,B)}:\M{d} &\to \M{d} \\ 
    X &\mapsto \includegraphics[scale=1, align=c]{Phi1-Z.pdf} = \operatorname{diag}(A\ket{\operatorname{diag}X}) + \widetilde{B}\,\odot X \\ 
    \Phi^{(3)}_{(A,B,C)}:\M{d} &\to \M{d} \\ 
    X &\mapsto \includegraphics[scale=1, align=c]{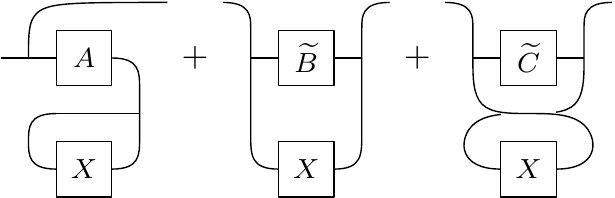} \\[0.2cm]
    &= \operatorname{diag}(A\ket{\operatorname{diag}X}) + \widetilde{B}\,\odot X + \widetilde{C}\,\odot X^\top,
\end{alignat*}
where $\odot$ denotes the entrywise matrix product and $\widetilde{Z}:= Z - \operatorname{diag}Z$.
\end{definition}

Note that in the above definition, we have simply extended the classical action of $\Phi_A$ by adding a specific non-diagonal action comprised of two well-known actions: matrix transposition and entrywise matrix multiplication. It turns out that these channels enjoy special diagonal unitary/orthogonal covariance properties, which we now illustrate. We denote the groups of diagonal unitary and diagonal orthogonal matrices in $\M{d}$ by $\mathcal{DU}_d$ and $\mathcal{DO}_d$, respectively.

\begin{theorem}\label{theorem:DOC-ABC}
Let $\Phi:\M{d}\to \M{d}$ be a linear map. Then, the following equivalences hold.
\begin{alignat*}{2}
    \forall X\in \M{d},\forall U\in\mathcal{DU}_d &: \,\, \Phi(UXU^\dagger) = U^\dagger \Phi(X) U   &&\iff \Phi = \Phi^{(1)}_{(A,C)}, \\ 
    \forall X\in \M{d}, \forall U\in\mathcal{DU}_d &: \,\, \Phi(UXU^\dagger) = U \Phi(X) U^\dagger &&\iff \Phi = \Phi^{(2)}_{(A,B)}, \\
     \forall X\in \M{d}, \forall O\in\mathcal{DO}_d &: \,\, \Phi(\,OXO\,) = O \Phi(X) O &&\iff \Phi = \Phi^{(3)}_{(A,B,C)},
\end{alignat*}
where each equivalence above holds for some $A,B,C\in \M{d}$ with $\operatorname{diag}A=\operatorname{diag}B=\operatorname{diag}C$.
\end{theorem}

Let us now define the Choi matrix of a linear map $\Phi:\M{d}\to \M{d}$ as 
\begin{equation}
    J(\Phi) := \sum_{1\leq i,j\leq d} \Phi(\ketbra{i}{j}) \otimes \ketbra{i}{j}.
\end{equation}
Then, for all $A,B,C\in \M{d}$ with $\operatorname{diag}A=\operatorname{diag}B=\operatorname{diag}C$, it is easy to show that
\begin{align}
    J(\Phi^{(1)}_{(A,C)}) &= \includegraphics[align=c]{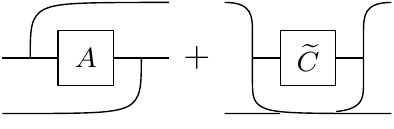} = \sum_{i,j=1}^d A_{ij} \ketbra{ij}{ij} + \sum_{1\leq i\neq j\leq d} C_{ij} \ketbra{ij}{ji} =: X^{(1)}_{(A,C)}, \label{eq:LDUI} \\[0.2cm]
    J(\Phi^{(2)}_{(A,B)}) &= \includegraphics[align=c]{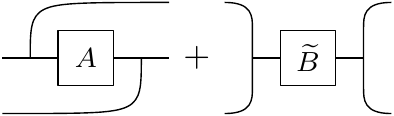} = \sum_{i,j=1}^d A_{ij} \ketbra{ij}{ij} + \sum_{1\leq i\neq j\leq d} B_{ij} \ketbra{ii}{jj} =: X^{(2)}_{(A,B)}, \label{eq:CLDUI} \\[0.2cm]
    J(\Phi^{(3)}_{(A,B,C)}) &= \includegraphics[align=c]{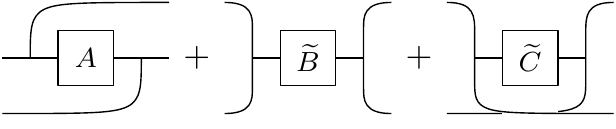} \nonumber \\ 
    &=\sum_{i,j=1}^d A_{ij} \ketbra{ij}{ij} + \sum_{1\leq i\neq j\leq d} B_{ij} \ketbra{ii}{jj} + \sum_{1\leq i\neq j\leq d} C_{ij} \ketbra{ij}{ji} =: X^{(3)}_{(A,B,C)}. \label{eq:LDOI}
\end{align}

These Choi matrices also inherit certain diagonal unitary and orthogonal invariance properties from their defining linear maps. We summarize all these symmetries in Table~\ref{tab:DOC}.

{\renewcommand{\arraystretch}{1}
\begin{table}[H]
    \centering
    \begin{tabular}{|r|l|c|c|} 
\hline
\emph{Acronym}    & \emph{Symmetry}                                & \emph{Linear map condition}  & \emph{Choi matrix condition}                                 \\ 
\hline\hline 
DUC & \specialcell[c]{Diagonal unitary \\ covariant} & \centered{ \hspace{0.1cm}$\Phi(UXU^\dagger) = U^\dagger \Phi(X) U$}     &  $(U \otimes U) J(\Phi) (U^\dagger \otimes U^\dagger) = J(\Phi)$   \\ 
\hline
CDUC &  \specialcell[c]{Conjugate diagonal unitary \\ covariant} & \hspace{0.1cm}$\Phi(UXU^\dagger) = U \Phi(X) U^\dagger$ & $(U \otimes \bar U) J(\Phi) (U^\dagger \otimes U^\top) = J(\Phi)$  \\ 
\hline
DOC  & \specialcell[c]{Diagonal orthogonal \\ covariant}      & $\Phi(\,OXO\,) = O \Phi(X) O $ & $(O\otimes O) J(\Phi) (O\otimes O) = J(\Phi)$   \\
\hline
\end{tabular}
    \caption{\centering Diagonal symmetries of a linear map $\Phi:\M{d}\to \M{d}$. The conditions above hold for all $X\in \M{d}, U\in \mathcal{DU}_d$, and $O\in \mathcal{DO}_d$.}
    \label{tab:DOC}
\end{table} }

Let $\T{d}$ be the vector space of all linear maps from $\M{d}\to \M{d}$. We will denote the vector subspaces of all DUC, CDUC, and DOC linear maps in $\T{d}$ by $\DUC_d, \CDUC_d,$ and $\DOC_d$, respectively. If we define the matrix spaces
\begin{align}
    \MLDUI{d} &\coloneqq \{ (A,B) \in \M{d}\times \M{d} \, \big| \, \operatorname{diag}(A)=\operatorname{diag}(B) \}, \text{ and } \label{eq:MLDUI}\\
    \MLDOI{d} &\coloneqq \{ (A,B,C) \in \M{d}\times \M{d}\times \M{d} \, \big| \, \operatorname{diag}(A)=\operatorname{diag}(B)=\operatorname{diag}(C)\}, \label{eq:MLDOI}
\end{align}
and equip them with component-wise addition and scalar multiplication, then Theorem~\ref{theorem:DOC-ABC} tells us that the following vector space isomorphisms hold.  
\begin{equation}
     \DUC_d \cong \CDUC_d \cong \MLDUI{d} \quad\text{and} \quad \DOC_d \cong \MLDOI{d}.
\end{equation}

\begin{remark}
$\DUC_d$ and $\CDUC_d$ form vector subspaces of $\DOC_d$. More precisely, 
\begin{equation*}
    \forall (A,B)\in \MLDUI{d}: \quad \Phi^{(3)}_{(A,\operatorname{diag}B,B)} = \Phi^{(1)}_{(A,B)} \,\,\,\text{ and }\,\,\, \Phi^{(3)}_{(A,B,\operatorname{diag}B)} = \Phi^{(2)}_{(A,B)}.
\end{equation*}
\end{remark}

The DUC, CDUC, and DOC classes of linear maps and the associated Choi matrices have been extensively studied in \cite{Singh2021diagonalunitary}. We note some important results about these maps below. Before proceeding further, let us note that the spectrum of a linear map $\Phi:\M{d}\to \M{d}$ is equal to the spectrum of its matrix representation in the standard basis of $\M{d}$ (we denote it by $M(\Phi)\in \M{d^2}\simeq \M{d}\otimes \M{d}$). Moreover, it is well-known and easy to check that $M(\Phi)$ is nothing but a realignment of the Choi matrix $J(\Phi)$:
\begin{equation}
    M(\Phi) = J(\Phi)^R,
\end{equation}
where the \emph{realignment} of a bipartite matrix $X$, denoted as $X^R$, is defined in coordinates as
$$\langle ij|X^R|kl\rangle = \langle ik|X|jl\rangle \qquad \forall i,j,k,l.$$

\begin{figure}[H]
    \centering
    \includegraphics[scale=1.2]{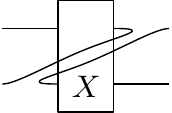}
    \caption{Visual depiction of a realigned bipartite matrix $X^R$.}
    \label{fig:PT+realign}
\end{figure}

\begin{theorem}\label{theorem:DOC-spec}
The spectrum of $\Phi^{(3)}_{(A,B,C)}\in \DOC_d$ admits the following expression:
\begin{equation*}
    \operatorname{spec}\Phi^{(3)}_{(A,B,C)} = \operatorname{spec}A \cup \left( \bigcup_{i<j} \operatorname{spec} \begin{bmatrix} B_{ij} & C_{ij} \\ C_{ji} & B_{ji} \end{bmatrix}   \right).
\end{equation*}
Moreover, all eigenmatrices of $\Phi^{(3)}_{(A,B,C)}$ associated with 
\begin{itemize}
    \item $\lambda\in\operatorname{spec}A$ are of the form $\operatorname{diag}\ket{v}$, where $\ket{v}$ are the eigenvectors of $A$ associated with $\lambda$.
    \item $\lambda\in \operatorname{spec} \begin{bmatrix} B_{ij} & C_{ij} \\ C_{ji} & B_{ji} \end{bmatrix}$ are of the form $v_1\ketbra{i}{j} + v_2\ketbra{j}{i}$, where 
    \begin{equation*}
        \begin{bmatrix} v_1 \\ v_2 \end{bmatrix} \text{ are the eigenvectors of } \begin{bmatrix} B_{ij} & C_{ij} \\ C_{ji} & B_{ji} \end{bmatrix} \text{ associated with } \lambda.
    \end{equation*}
\end{itemize}
\end{theorem}
\begin{proof}
The results follow immediately by noting that the following block decompositions hold:
\begin{align}
    J(\Phi^{(3)}_{(A,B,C)}) = X^{(3)}_{(A,B,C)} &= B \oplus \left( \bigoplus_{i < j} \begin{bmatrix} A_{ij} & C_{ij} \\ C_{ji} & A_{ji} \end{bmatrix} \right), \\
    M(\Phi^{(3)}_{(A,B,C)}) = J(\Phi^{(3)}_{(A,B,C)})^R = X^{(3)}_{(B,A,C)} &= A \oplus \left( \bigoplus_{i < j} \begin{bmatrix} B_{ij} & C_{ij} \\ C_{ji} & B_{ji} \end{bmatrix} \right),
\end{align}
see \cite[Proposition 4.1 and 4.3]{Singh2021diagonalunitary}.
\end{proof}

\begin{theorem}\label{theorem:DOC-cptp}
$\Phi^{(3)}_{(A,B,C)}\in \DOC_d$ is a quantum channel if and only if
\begin{itemize}
    \item $A$ is column stochastic,
    \item $B$ is positive semi-definite,
    \item $C=C^{\dagger}$, and $A_{ij}A_{ji}\geq |C_{ij}|^2$ for all $i<j$.
\end{itemize}
\end{theorem}
\begin{proof}
See \cite[Lemma 6.14]{Singh2021diagonalunitary}.
\end{proof}

Using the above theorem, we can split the action of a DOC channel on some input state $\rho$ into two parts. Firstly, we have a classical action on the diagonal probability distribution $\ket{\operatorname{diag}\rho}=\sum_i \rho_{ii}\ket{i}$ of $\rho$ that is facilitated by the `classical core' $A$. Secondly, we have a mixture of transposition with some entrywise product actions (performed by $B,C$). Several interesting and practically relevant classes of quantum channels have these kinds of actions: depolarizing and transpose-depolarizing channels, multi-level amplitude damping channels, Schur multipliers, to name a few. A more comprehensive list of examples can be found in \cite[Section 7]{Singh2021diagonalunitary}. 

Now that we have all the background covered, we can start talking about the ergodic properties of the recently introduced classes of channels. Recall that we wish to be able deduce all the said properties for a given DOC channel from its underlying `classical core', so that one can easily check if a given DOC channel has these properties by utilizing the classical results from Section~\ref{subsec:graphs-ergodic}. The expression for the spectrum of a DOC map given in Theorem~\ref{theorem:DOC-spec} will be crucial in realizing this goal. Hence, for $(B,C)\in \MLDUI{d}$ with $B=B^\dagger, C=C^\dagger$ and $i<j$, let us define
\begin{equation}
    \lambda^{\pm}_{ij} (B,C) := \frac 1 2 \left[ B_{ij}+B_{ji} \pm \sqrt{(B_{ij}-B_{ji})^2+4|C_{ij}|^2}\right],
\end{equation}
so that for all $i<j$, we have 
\begin{equation}
    \operatorname{spec}\begin{bmatrix} B_{ij} & C_{ij} \\ C_{ji} & B_{ji} \end{bmatrix}  = \{\lambda^+_{ij} (B,C), \lambda^-_{ij} (B,C) \}.
\end{equation}
Either by direct computation or by using the Gershgorin circle theorem \cite[Chapter 6]{Horn2012matrix}, it can be verified that
\begin{equation}\label{eq:lambdaij<=}
    \forall i<j: \quad |\lambda^\pm_{ij}(B,C)| \leq |B_{ij}| + |C_{ij}|.
\end{equation}

With the relevant notation in place, we can state the two main results of this section.

\begin{theorem}\label{theorem:DOC-ergodic/mixing}
A \emph{DOC} channel $\Phi^{(3)}_{(A,B,C)}:\M{d}\to \M{d}$ is 
\begin{itemize}
    \item ergodic if and only if $A$ is ergodic and $\lambda^{\pm}_{ij} (B,C) \neq 1$ for all $i<j$.
    \item mixing if and only if $A$ is mixing and $|\lambda^{\pm}_{ij} (B,C)| \neq 1$ for all $i<j$.
\end{itemize}
In both cases, the unique stationary state of $\Phi^{(3)}_{(A,B,C)}$ is $\operatorname{diag}\ket{\pi}= \sum_i \pi_i \ketbra{i} \in \St{d}$, where $\ket{\pi}\in \Delta^d$ is the unique stationary distribution of $A$.
\end{theorem}

\begin{proof}
Let us recall the expression for the spectrum of  $\Phi=\Phi^{(3)}_{(A,B,C)}$ (see Theorem~\ref{theorem:DOC-spec}):
\begin{equation*}
    \operatorname{spec}\Phi^{(3)}_{(A,B,C)} = \operatorname{spec}A  \cup \left( \bigcup_{i<j} \{ \lambda^{\pm}_{ij} (B,C) \}\right).
\end{equation*}
Note that $\lambda=1$ is always an eigenvalue of $A$ since it is column stochastic. Hence, it is clear that $\lambda=1$ is a simple eigenvalue of $\Phi$ if and only if it is a simple eigenvalue of $A$ and $\lambda^{\pm}_{ij} (B,C) \neq 1$ for all $i<j$. Moreover, $\lambda=1$ is the only peripheral eigenvalue of $\Phi$ if and only if it is the only peripheral eigenvalue of $A$ and $|\lambda^{\pm}_{ij} (B,C)| \neq 1$ for all $i<j$. The desired conclusions then immediately follow from the spectral definitions of ergodicity and mixing from Section~\ref{subsec:Qergodic}. The given form of the stationary state of $\Phi$ follows trivially from Theorem~\ref{theorem:DOC-spec}.
\end{proof}

\begin{theorem}\label{theorem:DOC-irred-prim}
A \emph{DOC} channel $\Phi^{(3)}_{(A,B,C)}:\M{d}\to \M{d}$ for $d\geq 3$ is 
\begin{itemize}
\item irreducible if and only if $A$ is irreducible.
\item primitive if and only if $A$ is primitive.
\end{itemize}
If $d=2$, then $\Phi^{(3)}_{(A,B,C)}$ is
\begin{itemize}
    \item irreducible if and only if $A$ is irreducible and $\lambda^{\pm}_{12} (B,C) \neq 1$.
    \item primitive if and only if $A$ is primitive and $|\lambda^{\pm}_{12} (B,C)| \neq 1$.
\end{itemize}
\end{theorem}

\begin{proof}
We deal with the irreducible case and leave the primitive case for the reader. Clearly, if $$\Phi = \Phi^{(3)}_{(A,B,C)}$$ is irreducible, then $A$ is irreducible and $\lambda^{\pm}_{ij} (B,C) \neq 1$ for all $i<j$ (see Theorem~\ref{theorem:DOC-ergodic/mixing}).

Conversely, assume that $A$ is irreducible so that $\lambda=1$ is a simple eigenvalue of $A$ and let $\ket{\pi}\in \Delta^d$ be the unique stationary distribution of $A$ with full support. Then, $\operatorname{diag}\ket{\pi}$ is a positive definite stationary state for $\Phi$. In order to show that $\Phi$ is irreducible, we must show that this is the unique stationary state of $\Phi$, i.e., $\lambda=1$ is a simple eigenvalue of $\Phi$. Assume that this is not the case and there exist $i<j$ such that $\lambda^\pm_{ij}(B,C) = 1$. This implies that 
\begin{align}
    1 = |\lambda^\pm_{ij}(B,C)| \leq |B_{ij}| + |C_{ij}| &\leq \sqrt{A_{ii}A_{jj}} + \sqrt{A_{ij}A_{ji}} \nonumber \\
    &\leq \frac{1}{2} (A_{ii}+A_{jj}+A_{ij}+A_{ji}) \leq 1.
\end{align}
The first, second, and third inequalities above follow respectively from Eq.~\eqref{eq:lambdaij<=}, Theorem~\ref{theorem:DOC-cptp}, and the arithmetic-geometric mean inequality. The final inequality follows from the fact that $A$ is column stochastic (see Theorem~\ref{theorem:DOC-cptp}). Hence, we must have $A_{ii}+A_{ji} = A_{jj}+A_{ij} = 1$. In other words, the subset $\{i,j\}$ forms a closed class of the vertex set of $G_A$ (see Remark~\ref{remark:closed-class}). Now, if $d\geq 3$, there exists $k \notin \{i,j\}$ such that neither $i$ nor $j$ lead to $k$. This violates the strong connectivity of $G_A$ (or equivalently, the irreducibility of $A$ (see Theorem~\ref{theorem:perron})) and thus leads to a contradiction. Hence, if $d\geq 3$, then $A$ being irreducible implies that $\Phi$ is irreducible as well. If $d=2$, however, the above reasoning fails and we must additionally impose the constraint $\lambda^\pm_{12}(B,C)\neq 1$ on top of irreducibility of $A$ to ensure that $\Phi$ is irreducible as well.
\end{proof}

\begin{remark}\label{remark:periphery-DOC-A}
The proof of Theorem~\ref{theorem:DOC-irred-prim} also works to show that the peripheral eigenvalues of an irreducible \emph{DOC} channel $\Phi^{(3)}_{(A,B,C)}\in \T{d}$ (for $d\geq 3$) are precisely the peripheral eigenvalues of $A$. This is not true when $d=2$ (see Example~\ref{eg:3}).
\end{remark}

Let us now illustrate the above results with some examples.

\begin{example}
Let $(A,B,C)\in \MLDOI{3}$ be defined as
\begin{equation}\label{eq:eg1}
     A = \left(\begin{array}{c c c}
   0.5 & 0.5 & 0.5 \\
   0.5 & 0.5 & 0.5  \\
   0 & 0 & 0
\end{array} \right), \,\, B = C = \left(\begin{array}{c c c} 
   0.5 & x & 0 \\
   x & 0.5 & 0 \\
   0 & 0 & 0
\end{array} \right),
\end{equation}
where $x\in \mathbb R$ is such that $|x|\leq 0.5$. Clearly, $G_A$ is not strongly connected but has a unique closed and aperiodic class. Consequently, $A$ is mixing according to Corollary~\ref{corollary:Cmixing}. However, the channel $\Phi=\Phi^{(3)}_{(A,B,C)}$ is not ergodic if $x=0.5$, since then $\lambda^+_{12}(B,C)=1$. If $|x|<0.5$, then $|\lambda^\pm_{12}(B,C)|<1$ and $\Phi$ becomes mixing. If $x=-0.5$, then $\lambda^{\pm}_{12}(B,C)\neq 1$ but $\lambda^-_{12}(B,C)=-1$. Hence, $\Phi$ is no longer mixing but is ergodic. 
\begin{figure}[H]
    \centering
    \includegraphics[scale=1.4]{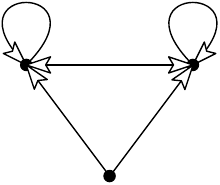}
    \caption{$G_A$ for $A$ given in Eq.~\eqref{eq:eg1}.}
    \label{fig:eg1}
\end{figure}
\end{example}
On the other hand, if $A$ is irreducible (resp. primitive), then $\Phi$ must be irreducible (resp. primitive) according to Theorem~\ref{theorem:DOC-irred-prim}. For instance, let $(A,B,C)\in \MLDOI{3}$ be defined as
\begin{equation}\label{eq:eg11}
     A = \left(\begin{array}{c c c}
   0.2 & 0.4 & 0.4 \\
   0.4 & 0.2 & 0.4  \\
   0.4 & 0.4 & 0.2
\end{array} \right), \,\, B = C = \left(\begin{array}{c c c} 
   0.2 & * & * \\
   * & 0.2 & * \\
   * & * & 0.2
\end{array} \right).
\end{equation}
Then, $A$ is clearly primitive (because $G_A$ is aperiodic, see Theorem~\ref{theorem:perron}) and hence $\Phi$ must be primitive too regardless of what $B$ and $C$ matrices we choose. 
\begin{figure}[H]
    \centering
    \includegraphics[scale=1.4]{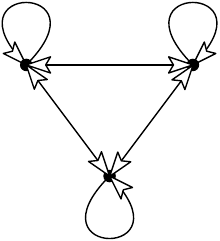}
    \caption{$G_A$ for $A$ given in Eq.~\eqref{eq:eg11}}
    \label{fig:eg11}
\end{figure}

\begin{example}
Let $(A,B,C)\in \MLDOI{2}$ be defined as
\begin{equation}\label{eq:eg2}
     A = B = C = \left(\begin{array}{c c}
   0.5 & 0.5 \\
   0.5 & 0.5
\end{array} \right).
\end{equation}
In this case, even though $G_A$ is strongly connected (and aperiodic) meaning that $A$ is primitive, $\Phi=\Phi^{(3)}_{(A,B,C)}$ is not ergodic since $\lambda^+_{12}(B,C)=1$. As before, $\Phi$ can be made primitive by reducing the magnitude of the off-diagonal entries of $B,C$, so that $|\lambda^\pm_{12}(B,C)|<1$.
\begin{figure}[H]
    \centering
    \includegraphics[scale=1.4]{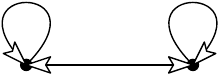}
    \caption{$G_A$ for $A$ given in Eq.~\eqref{eq:eg2}.}
\end{figure}
When $d\geq 3$, such an examples do not exist since Theorem~\ref{theorem:DOC-irred-prim} guarantees that $\Phi$ is irreducible (resp. primitive) if and only if $A$ is irreducible (resp. primitive).
\end{example}

\begin{example}\label{eg:3}
Let $(A,B,C)\in \MLDOI{2}$ be defined as
\begin{equation}
     A = \left(\begin{array}{c c}
   0.5 & 0.5 \\
   0.5 & 0.5
\end{array} \right), \,\, B = \left(\begin{array}{c c}
   \phantom{-}0.5 & -0.5 \\
   -0.5 & \phantom{-}0.5
\end{array} \right) = C.
\end{equation}
It is easy to check that $\Phi=\Phi^{(3)}_{(A,B,C)}$ is irreducible since $A$ is irreducible and $\lambda^{\pm}_{12}(B,C)\neq 1$. Note, however, that the peripheral spectrum of $\Phi$ is not the same as that of $A$. Since $A$ is primitive, $\lambda=1$ is its only peripheral eigenvalue. In contrast, $\Phi$ has an additional peripheral eigenvalue $\lambda^-_{12}(B,C) = -1$. This phenomenon only occurs in $d=2$ (see Remark~\ref{remark:periphery-DOC-A}).
\end{example}

\section{Lattice models}\label{sec:lattice-models}
In what follows, we will denote the unitary groups in $\M{d}$ and $\M{d}\otimes \M{d}$ by $\U{d}$ and $\UU{d}$, respectively. We consider one dimensional spin chains consisting of $2L$ $(L\in \mathbb{N})$ qudit lattice sites, each comprising of a $d$-dimensional Hilbert space $\C{d}$. We label the sites using integers from the set $Z_L := \{-L+1,-L+2, \ldots ,L-1, L \}$. Time evolution of the system is discrete and is governed by a staggered or brickwork unitary circuit of the form 
\begin{equation}\label{eq:circuitmodel}
    \mathbb{U}(t) = \underbrace{\mathbb{U}_{\pm} \ldots \ldots \mathbb{U}_+\mathbb{U}_-\mathbb{U}_+ \mathbb{U}_-}_t, \qquad t\in\mathbb{N},
\end{equation}
where 
\begin{align}
    \mathbb{U}_- &= U_{(-L+2,-L+3)}U_{(-L+4,-L+5)}\ldots U_{(L-2,L-1)} U_{(L,-L+1)}, \\
    \mathbb{U}_+ &= U_{(-L+1,-L+2)}U_{(-L+3,-L+4)}\ldots U_{(L-3,L-2)} U_{(L-1,L)},
\end{align}
and for $x,y\in Z_L$, the unitary operator $U_{(x,y)}$ acts by applying $U\in \UU{d}$ at the $x,y$ lattice sites and identity everywhere else. Adjacent sites are coupled together via a bipartite unitary gate $U$ and periodic boundary conditions are imposed by coupling the $L^{th}$ and $(-L+1)^{th}$ sites together. Figure~\ref{fig:U(4)} shows a slice of the time evolution unitary operator $\mathbb{U}(t)$ for $t=4$ time steps and $2L \gg t$. It is perhaps worthwhile to emphasize that the entire global unitary evolution operator is built from a single local building block $U\in \UU{d}$. 

\begin{figure}[httb!]
    \centering
    \includegraphics[scale=1.2]{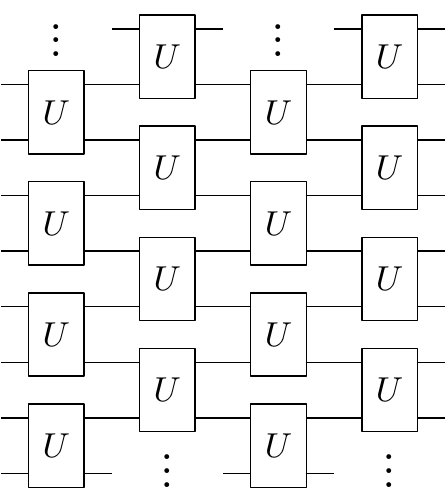}
    \caption{Time evolution operator $\mathcal{U}(t)$ for $t=4$ time steps and the number of lattice sites $2L \gg t$. Above, the time coordinate $t$ is depicted horizontally, and counts the number of layers in the circuit. The spatial dimension $x$ is depicted vertically, and corresponds to the width of the circuit.}
    \label{fig:U(4)}
\end{figure}
For a local operator $A\in \M{d}$ and lattice site $x\in Z_L$, we define
\begin{equation}
    A_x := \iden_d \otimes \ldots \otimes \iden_d \otimes \vertarrowbox[3ex]{A}{x} \otimes \iden_d \otimes \ldots \otimes \iden_d \in \M{d}^{\otimes 2L}.
\end{equation}
Given a lattice site $x\in Z_L$, the local two-point correlation sequence (with respect to the quantum state $\rho = d^{-2L}\iden$) between $A,B\in \M{d}$ is defined as
\begin{equation}
    C_{A,B} (x,t) := \operatorname{Tr}\left( \mathbb{U}^\dagger (t) A_0 \, \mathbb{U}(t) B_x \right) - \operatorname{Tr}(A_0)\operatorname{Tr}\left( \frac{B_x}{d^{2L}} \right), \quad t\in\mathbb{N}.
\end{equation}

The readers should compare this definition with the ones given in Definition~\ref{def:corr_measure} and Remark~\ref{remark:corr_Q}. It turns out that there is a natural light cone structure built into our model, which gets manifested in the following manner. Just by exploiting unitarity of the building block $U\in \UU{d}$, it can be shown that the correlations $C_{A,B}(x,t)$ vanish for all $A,B\in \M{d}$ whenever $|x|>t$. Moreover, the correlations present on the edge of the light cone can be computed explicitly \cite{Bruno2019dual}:
\begin{align}
    \forall t\in \mathbb{N}: \qquad C_{A,B}(\pm t,t) &= d^{2L-1}\operatorname{Tr}\left(\Lambda^t_\pm(U) (A)B \right) - \operatorname{Tr}(A_0)\operatorname{Tr}\left( \frac{B_x}{d^{2L}} \right) \nonumber \\
    &= d^{2L-1} \left[ \operatorname{Tr}\left(\Lambda^t_\pm(U) (A)B \right) - \operatorname{Tr}(A)\operatorname{Tr} \left( \frac{B}{d} \right)  \right]
\end{align}
where $\Lambda_\pm (U) :\M{d}\to \M{d}$ are unital quantum channels defined in Fig~\ref{fig:Lambda+-}. Note that the expression in the square brackets is nothing but the correlation sequence of $A,B\in \M{d}$ with respect to the fixed state $\iden_d/d$ and the $\Lambda_\pm(U)$ channels as defined in Remark~\ref{remark:corr_Q}. In other words, the long term behaviour of the correlations on the edge of the light cone are solely determined by the ergodic properties of the $\Lambda_\pm(U)$ channels. But what about the correlations inside the light cone? When $|x|<t$, it can be shown that the number of instances of the unitary $U$ in the expression $C_{A,B}(x,t)$ grows exponentially with the distance $t-x$ from the edge of the light cone, thus making explicit computations possible only for short times.

\begin{figure}[H]
    \centering
    \includegraphics[scale=1.3]{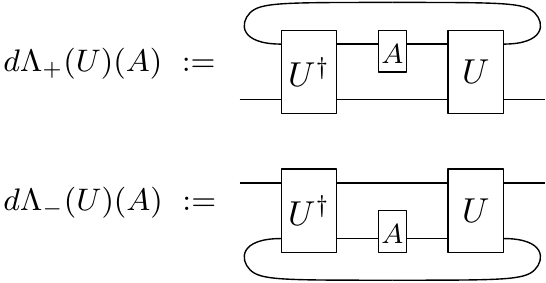}
    \caption{The action of $\Lambda_\pm$ channels on $\M{d}$.}
    \label{fig:Lambda+-}
\end{figure}

We now recall the following important notion of \emph{dual} unitarity \cite{Bruno2019dual,gopalakrishnan2019unitary}. 

\begin{definition}
    A bipartite unitary matrix $U \in \UU{d}$ is called \emph{dual} if its realignment, $U^R$, is also a unitary matrix.
\end{definition}

Going back to our setting, if we additionally impose the constraint of dual unitarity on $U$, then the spatial and temporal directions switch roles in Figure~\ref{fig:U(4)}, so that our global evolution operator becomes unitary also in the spatial direction. Consequently, the correlations $C_{A,B}(x,t)$ vanish for $|x|<t$ as well \cite{Bruno2019dual}. Hence, for such `dual unitary' circuit models, the only non-trivial correlations are present on the edge of the light cone, and these are determined by the $\Lambda_\pm(U)$ channels defined above. In this paper, we will only focus on the $\Lambda_+$ channels. The analysis can easily be replicated for $\Lambda_-$ channels as well.

\begin{figure}[H]
    \centering
    \includegraphics[scale=1.4]{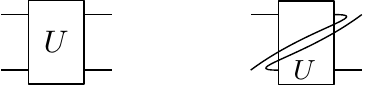}
    \caption{A unitary operator $U$ (left) is called \emph{dual} if its realignment $U^R$ (right) is also unitary.}
    \label{fig:dualU}
\end{figure}

Motivated by the above discussion, we can introduce the following definition. Recall that for unital channels, ergodicity and mixing properties are equivalent to irreducibility and primitivity, respectively (see Remark~\ref{remark:ergodic-unital}).

\begin{definition}\label{def:circuit-ergodic}
Let $\{\mathbb{U}(t) \}_{t\in \mathbb{N}}$ be a unitary circuit acting on a $1$-\emph{D} lattice consisting of $2L$ qudits and constructed from a unitary $U\in \UU{d}$ via Eq.~\eqref{eq:circuitmodel}. We say that $\{\mathbb{U}(t) \}_{t\in \mathbb{N}}$ is 
\begin{itemize}
    \item \emph{non-interacting} if $\Lambda_+(U):\M{d}\to \M{d}$ is the identity channel,
    \item \emph{ergodic} if $\Lambda_+(U):\M{d}\to \M{d}$ is irreducible,
    \item \emph{mixing} if $\Lambda_+(U):\M{d}\to \M{d}$ is primitive.
    \item \emph{Bernoulli} if $\Lambda_+(U):\M{d}\to \M{d}$ is the completely depolarizing channel:
    \begin{equation}
        \forall A\in \M{d}: \quad \Lambda_+(U)(A) = \operatorname{Tr}(A)\frac{\mathbb{1}}{d}.
    \end{equation}
\end{itemize}
\end{definition}

\begin{remark}\label{remark:bernoulli}
Bernoulli circuits as defined above can be 
considered as extreme versions of mixing circuits, in the sense that the correlations $C_{A,B}(t,t)$ in such circuits vanish instantly, i.e., for all $A,B\in \M{d}$, $C_{A,B}(t,t)=0$ for $t\geq 1$. It was shown in \cite{arul2021dual} that a circuit $\{\mathbb{U}(t) \}_{t\in \mathbb{N}}$ is Bernoulli if and only if the underlying bipartite gate $U\in \UU{d}$ is \emph{perfect}, i.e., both the realigned operator $U^R$ and the partially transposed operator $U^\Gamma := (\operatorname{id} \otimes \operatorname{transp})(U)$ are unitary.
\end{remark}

\begin{figure}[H]
    \centering
    \includegraphics[scale=1.3]{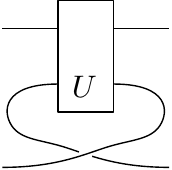}
    \caption{Visual depiction of a partially transposed bipartite matrix $U^\Gamma.$}
    \label{fig:my_label}
\end{figure}

\begin{remark}\label{remark:const-modes} 
In addition to Definition~\ref{def:circuit-ergodic}, we will also be interested in computing the number of peripheral eigenvalues of $\Lambda_+(U)$. The (traceless) eigenoperators $A\in \M{d}$ associated with the non-unit peripheral eigenvalues are the `non-decaying' (and non-constant) modes of the circuit, i.e., for all $B\in \M{d}$, $C_{A,B}(t,t)$ may fluctuate but $|C_{A,B}(t,t)|$ stays constant for all $t\in \mathbb{N}$. On the other hand, eigenoperators $A\in \M{d}$ associated with the unit eigenvalue are the `constant' modes of the circuit, i.e., for all $B\in \M{d}$, $C_{A,B}(t,t)$ stays constant for all $t\in \mathbb{N}$. 
\end{remark}

We have seen that explicit computation of two-point correlation functions between local operators is possible in certain brickwork dual unitary circuit models. However, a complete description of the family of dual unitary matrices in $\UU{d}$ is currently unavailable, although several different constructions of such operators have been proposed \cite{arul2021dual, Claeys2021dual,borsi2022remarks}\footnote{For the qubit case $d=2$, a full parametrization of the class of dual unitary matrices is given in \cite{Bruno2019dual}.}. In this regard, two authors of this paper came up with an interesting family of dual unitary matrices having certain local diagonal symmetries \cite{singh2022diagonal} as described in Table~\ref{tab:DOC}, which we again list below in Table~\ref{tab:LDOI}.
\begin{table}[H]
    \centering
    {\renewcommand{\arraystretch}{1.2}
\begin{tabular}{|r|l|c|} 
\hline
\emph{Acronym}    & \emph{Symmetry}                                   & \emph{Condition}                                        \\ 
\hline\hline
LDUI & local diagonal unitary invariant           & $(U \otimes U) X (U^\dagger \otimes U^\dagger) = X$          \\ 
\hline
CLDUI & conjugate local diagonal unitary invariant & $(U \otimes \bar U) X (U^\dagger \otimes U^\top) = X$  \\ 
\hline
LDOI  & local diagonal orthogonal invariant        & $(O \otimes O) X (O \otimes O) = X$    \\
\hline
\end{tabular}
}
    \caption{\centering Local diagonal symmetries of a bipartite matrix $X\in \M{d}\otimes \M{d}$. The conditions above hold for all $U\in \mathcal{DU}_d$ and $O\in \mathcal{DO}_d$.}
    \label{tab:LDOI}
\end{table}
Recall that the explicit forms of LDUI, CLDUI, and LDOI matrices in terms of the associated matrix triples $(A,B,C)\in \MLDOI{d}$ are given in Eqs.~\eqref{eq:LDUI}, \eqref{eq:CLDUI}, and \eqref{eq:LDOI}, respectively. With the appropriate definitions in hand, we can state a few main results from \cite{singh2022diagonal}.
\begin{proposition}\label{prop:unitary-condition}
An \emph{LDOI} matrix $X^{(3)}_{(A,B,C)}$ is unitary if and only if
\begin{itemize}
    \item $B$ is unitary,
    \item $\forall i < j$, there exists a phase $\omega_{ij} \in \mathbb T$ such that $A_{ji} = \omega_{ij} \overline{A_{ij}} \text{ and } C_{ji} = - \omega_{ij} \overline{C_{ij}},$
    \item $\forall i < j, \,\, |A_{ij}|^2 + |C_{ij}|^2 = 1$.
\end{itemize}
\end{proposition}
\begin{proof}
See \cite[Proposition 3.1]{singh2022diagonal}.
\end{proof}

\begin{proposition}\label{prop:dual}
An \emph{LDOI} matrix $X^{(3)}_{(A,B,C)}$ is dual unitary if and only if
\begin{itemize}
    \item $A$ and $B$ are unitary,
    \item $\forall i<j$, $|A_{ij}|^2 = |B_{ij}|^2 = 1-|C_{ij}|^2$,
    \item $\forall i<j$, there exist complex phases $\omega_{ij}\in \mathbb T$ such that
    $$A_{ji} = \omega_{ij} \overline{A_{ij}}, \qquad B_{ji} = \omega_{ij} \overline{B_{ij}}, \qquad C_{ji} = -\omega_{ij} \overline{C_{ij}}.$$
\end{itemize}
\end{proposition}
\begin{proof}
See \cite[Proposition 4.2]{singh2022diagonal}.
\end{proof}

Even though Proposition~\ref{prop:dual} fully characterizes the class of dual unitary LDOI matrices, it is difficult to explicitly find matrix triples $(A,B,C)\in \MLDOI{d}$ that satisfy the constraints of the proposition. Nevertheless, we can work with the following two families of triples: 

\begin{example}
\label{ex:dual1} 
Let $C\in\M{d}$ be an arbitrary phase matrix (i.e., $C_{ij}\in\mathbb{T} \,\,\forall i,j$) and $A=B=\operatorname{diag}C$. It can be shown that this family of triples corresponds precisely to the class of LDUI dual unitary matrices \cite[Proposition 4.2]{singh2022diagonal}.
\end{example}

\begin{example} \label{ex:dual2}
For any orthogonal projection $P\in \M{d}$, let $A=B=2P-\iden_d$. Construct $C\in \M{d}$ by fixing $\operatorname{diag}C=\operatorname{diag}A =\operatorname{diag}B$ and choose arbitrary $C_{ij}\in \mathbb C$ satisfying $|C_{ij}|^2 = 1 - |A_{ij}|^2 = 1- |B_{ij}|^2$ for all $i<j$. The remaining entries must then be $C_{ji}=-\overbar{C_{ij}}$ for all $i<j$ according to Proposition~\ref{prop:dual}.
\end{example}

In the remaining part of this section, we will use the above two families of LDOI dual unitary matrices to construct brickwork unitary circuits $\{\mathbb{U}(t) \}_{t\in \mathbb{N}}$ via Eq.~\eqref{eq:circuitmodel} and analyse their ergodic behaviour in accordance with Definition~\ref{def:circuit-ergodic} and Remark~\ref{remark:const-modes}. Let us first note that since no perfect LDOI unitary matrices exist \cite[Proposition 4.4]{singh2022diagonal}, any LDOI dual unitary brickwork circuit $\{\mathbb{U}(t) \}_{t\in \mathbb{N}}$ \emph{cannot} be Bernoulli, see Remark~\ref{remark:bernoulli}. Thus, our aim is to show that LDOI dual unitary brickwork circuits can exhibit all kinds of ergodic behaviours except Bernoulli. Similar analysis has been performed for the qubit case in \cite{Bruno2019dual}, and for higher dimensions in \cite{Claeys2021dual}. In fact, the construction of dual unitary matrices in \cite{Claeys2021dual} yields the class of LDUI dual unitaries from Example~\ref{ex:dual1}. This family has also been independently studied in \cite{arul2021dual}.

To begin with, let us note a crucial lemma, which shows that for an LDOI dual unitary gate $U$, the $\Lambda_+(U)$ channel inherits the diagonal orthogonal symmetry of $U$ and becomes DOC. 

\begin{lemma}\label{lemma:Lambda+LDOI}
Given matrices $(A,B,C)\in \MLDOI{d}$, define a new triple $(\mathcal{A},\mathcal{B},\mathcal{C})\in \MLDOI{d}$ by
\begin{align*}
\mathcal{A} &= d^{-1} \left[ A^\top\hspace{-0.12cm}\odot A^\dagger + B^\top\hspace{-0.12cm} \odot B^\dagger + \operatorname{diag}(\overbar C C^\top - 2C\odot \overbar C)\right], 
\\ \mathcal{B} &= d^{-1} \overbar C C^\top ,
\\ \mathcal{C} &= d^{-1} \left[A\,\odot \,B^\dagger + A^\dagger\odot B + \operatorname{diag}(\overbar C C^\top - 2C\odot \overbar C) \right].
\end{align*}
Then, for an \emph{LDOI} unitary matrix $X^{(3)}_{(A,B,C)}$, the following relation holds:
\begin{equation}
    \Lambda^+_{(A,B,C)} := \Lambda_+ \left( X^{(3)}_{(A,B,C)} \right)  = \Phi^{(3)}_{(\mathcal{A},\mathcal{B},\mathcal{C})} \in \DOC_d.
\end{equation}
\end{lemma}
\begin{proof}
We can compute the Choi matrix of $\Lambda^+_{(A,B,C)}$ as follows:
\begin{align*}
    d J(\Lambda^+_{(A,B,C)}) &= \includegraphics[scale=1.3,align=c]{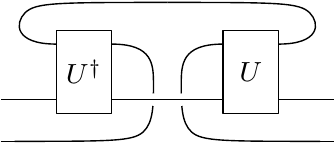} \\[0.2cm]
    &= \includegraphics[scale=1.3,align=c]{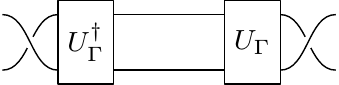} = F(U^\dagger_\Gamma U_\Gamma)F,
\end{align*}
where $U=X^{(3)}_{(A,B,C)}$, $U_\Gamma = (\operatorname{transp} \otimes \, \text{id})(U)$ is obtained from $U$ via partial transposition, and $F$ is the \emph{flip} (or \emph{swap}) operator in $\M{d}\otimes \M{d}$. Continuing our calculation, we obtain
\begin{align*}
    F(U^\dagger_\Gamma U_\Gamma)F &= F(X^{(3)}_{\bar{A},\bar{C},\bar{B}} X^{(3)}_{(A,C^\top,B^\top)})F = d X^{(3)}_{(\mathcal{A},\mathcal{B},\mathcal{C})} = d J(\Phi^{(3)}_{(\mathcal{A},\mathcal{B},\mathcal{C})}),
\end{align*}
where we have used \cite[Proposition 4.3]{Singh2021diagonalunitary} to compute $U_\Gamma$, $U_\Gamma^\dagger$ and \cite[Lemma 2.9]{singh2022diagonal} to compute the product $U^\dagger_\Gamma U_\Gamma$.
\end{proof}

In light of Lemma~\ref{lemma:Lambda+LDOI}, Theorem~\ref{theorem:DOC-irred-prim} and Theorem~\ref{theorem:perron}, checking the ergodic properties of an LDOI brickwork circuit is equivalent to analyzing the connectivity properties of the digraph $G_\mathcal{A}$, where $\mathcal{A}$ is the stochastic matrix from Lemma~\ref{lemma:Lambda+LDOI}. Moreover, when $G_\mathcal{A}$ is strongly connected, the number of non-decaying modes of the circuit (i.e., the number of peripheral eigenvalues of $\Lambda^+_{(A,B,C)}$) is equal to the number of peripheral eigenvalues of $\mathcal{A}$ (see Remark~\ref{remark:periphery-DOC-A}). If $G_\mathcal{A}$ is not strongly connected, then one can count the number of non-decaying modes of the circuit by explicitly analyzing the spectrum (see Theorem~\ref{theorem:DOC-spec}):
\begin{align}\label{eq:spec-LDOIdualcircuit}
    \operatorname{spec}\Lambda^+_{(A,B,C)} &= \operatorname{spec}\mathcal{A} \cup \left( \bigcup_{i<j} \operatorname{spec} \begin{bmatrix} \mathcal B_{ij} & \mathcal C_{ij} \\ \mathcal C_{ji} & \mathcal B_{ji} \end{bmatrix}   \right).
\end{align}

The next few subsections are dedicated to this analysis. We will consider circuits with underlying local dimension $d\geq 3$.

\subsection{Non-ergodic behaviour} \label{subsec:non-ergodic} In this subsection, we construct our brickwork circuits from the LDUI dual unitary matrices from Example~\ref{ex:dual1}. So, let $C\in \M{d}$ be a phase matrix and $A=B=\operatorname{diag}C$. Then, it can be easily checked that $\mathcal{A} = \mathcal{C} = \iden_d$, so that
\begin{equation}\label{eq:schur}
    \forall X\in \M{d}: \quad \Lambda^+_{(A,B,C)}(X) = \mathcal{B}\odot X = \frac{\overbar{C}C^\top \odot X }{d}
\end{equation}
is just a Schur multiplication channel. The graph $G_{\mathcal{A}}$ simply consists of $d$ isolated vertices with no edges amongst them, except self-loops. All diagonal matrix units $\{\ketbra{i}\}_{1\leq i\leq d}$ are constant modes of the circuit, so that the circuit is \emph{non-ergodic}. Since $C$ is a phase matrix, the Cauchy-Schwarz inequality implies that for all $k\neq l$,
\begin{equation}
     |\mathcal{B}_{kl}| = \frac{1}{d} |\langle C_{\text{row}\, k}, C_{\text{row}\, l} \rangle| \leq 1.
\end{equation}
Note that we use $C_{\text{row}\, k}$ to denote the $k^{th}$ row of $C\in \M{d}$. Thus, the remaining $d^2-d$ eigenvalues coming from the entries of $\mathcal{B}$ can be tuned to either lie on the periphery $\mathbb{T}$ or inside as follows:
\begin{itemize}
    \item We can choose $C_{\text{row}\, k}= \theta_{kl} C_{\text{row}\, l}$ for $k< l$ and $\theta_{kl}\in\mathbb T$, so that
    \begin{equation}
        \mathcal{B}_{kl} = \frac{1}{d}\langle C_{\text{row}\, k}, C_{\text{row}\, l} \rangle = \overbar{\theta_{kl}} = \overbar{\mathcal{B}_{lk}}.
    \end{equation}
    By choosing $\theta_{kl}=1$ for all $k<l$, we get $d^2-d$ additional constant modes $\{\ketbra{k}{l} \}_{1\leq k\neq l\leq d}$, so that our circuit becomes \emph{non-interacting}. If $\theta_{kl}\neq 1$ for some $k<l$, then we get two non-decaying and non-constant modes $\ketbra{k}{l}$ and $\ketbra{l}{k}$ via the non-unit peripheral eigenvalues $\theta_{kl}$ and $\overbar{\theta_{kl}}$, and our circuit becomes \emph{interacting}.
    \item In contrast, if $C_{\text{row}\, k}$ and $C_{\text{row}\, l}$ are chosen to be linearly independent for $k< l$, then the modes $\ketbra{k}{l}$ and $\ketbra{l}{k}$ decay, since the associated eigenvalues
    \begin{equation}
        |\mathcal{B}_{kl}| = \frac{1}{d} |\langle C_{\text{row}\, k}, C_{\text{row}\, l} \rangle| = |\mathcal{B}_{lk}| <1
    \end{equation}
    become non-peripheral.
\end{itemize}
Hence, by suitably choosing the rows of $C$ to be linearly independent (or not), we can tune the number of decaying/non-decaying modes in the circuit. However, note again that we always have at least $d$ constant modes in these circuits. One can easily construct non-ergodic circuits with less than $d$ constant modes by choosing the underlying dual unitary gates to be a direct sum of the LDUI dual unitary gates from this subsection and the LDOI dual gates from the next subsection. We leave the details of this construction as an exercise for the reader.

\subsection{Ergodic and mixing behaviour}
In this subsection, we use the full LDOI dual unitary family from Example~\ref{ex:dual2} to construct our circuits. So let $A=B=2P-\iden_d$ for some orthogonal projection $P\in \M{d}$ and we construct $C$ as in Example~\ref{ex:dual2}. If we simply choose $P$ to be a random orthogonal projection (say $P=VV^\dagger$ for a Haar random isometry $V\in \M{d\times d'}$), then all the entries of $A=B$ are non-zero almost surely, meaning that $\mathcal{A}$ is entrywise positive almost surely and hence $G_\mathcal{A}$ is aperiodic (i.e., $A$ is primitive (see Theorem~\ref{theorem:perron})) almost surely. Hence, the channel $\Lambda^+_{(A,B,C)}$ is primitive and our circuit generically exhibits ergodic and mixing behaviour, with the maximally mixed state $\iden_d/d$ being the unique constant mode of the circuit.

\subsection{Ergodic and non-mixing behaviour}
We again revert back to the family from Example~\ref{ex:dual1} where $A=B=\operatorname{diag}C$ for some phase matrix $C\in \M{d}$. We saw that this gives a Schur multiplication channel $\Lambda^+_{(A,B,C)}$ (see Eq.~\eqref{eq:schur}) which has at least $d$ constant modes. We now break this degeneracy by changing the underlying dual unitary building block to
\begin{equation}
U = (\pi\otimes \iden_d) X^{(3)}_{(A,B,C)},
\end{equation}
where $\pi\in \M{d}$ is a cyclic permutation matrix defined on the standard basis as $\pi^\dagger \ket{i}=\ket{i+1}$ (mod d) for $i\in \{1,2,\ldots d \}$. Since local transformations do not affect the dual unitary property of matrices, $U$ is again dual unitary. Moreover, it is easy to check from Figure~\ref{fig:Lambda+-} that
\begin{align}
    \forall X\in \M{d}: \quad \Lambda_+(U)(X) &= \Lambda^+_{(A,B,C)}(\pi^\dagger X \pi) \nonumber \\
    &= \frac{\overbar{C}C^\top \odot \pi^\dagger X \pi  }{d} .
\end{align}
The action of $\Lambda_+(U)$ on the matrix units takes the following permutation like form:
\begin{equation}\label{eq:cycles}
    \Lambda_+(U)(\ketbra{i}{j}) = \mathcal{B}_{i+1,j+1} \ketbra{i+1}{j+1} \quad \text{(mod d)}.
\end{equation}
Clearly, the diagonal matrix units $\{\ketbra{i} \}_{1\leq i\leq d}$ which were previously the constant modes of the circuit are now related in a cyclic fashion (see Figure~\ref{fig:dualcycle-diag}), so that the previously $d$-fold degenerate unit eigenvalue splits into $d$ distinct peripheral eigenvalues $\{e^{2\pi ik/d}\}_{k\in \{1,2,\ldots ,d\}}$. This ensures that the circuit cannot be mixing. 
\begin{figure}[H]
    \centering
    \includegraphics[scale=0.7]{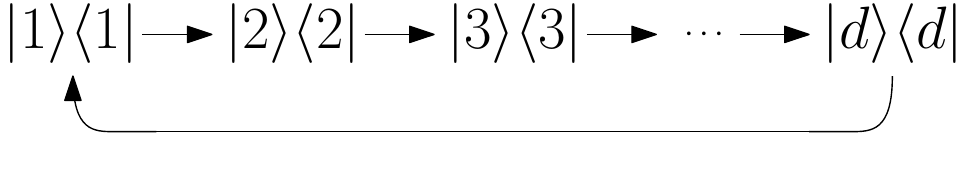}
    \caption{Cyclic permutation of the diagonal matrix units under $\Lambda_+(U)$.}
    \label{fig:dualcycle-diag}
\end{figure}
To ensure ergodicity, we now show that the remaining eigenvalues can be made to lie inside the periphery, so that our circuit finally becomes \emph{ergodic} and \emph{non-mixing}. Firstly, observe that the off-diagonal matrix units also decompose into $d-1$ disjoint cyclic orbits under the action of $\Lambda_+(U)$.
\begin{figure}[H]
    \centering
    \includegraphics[scale=0.7]{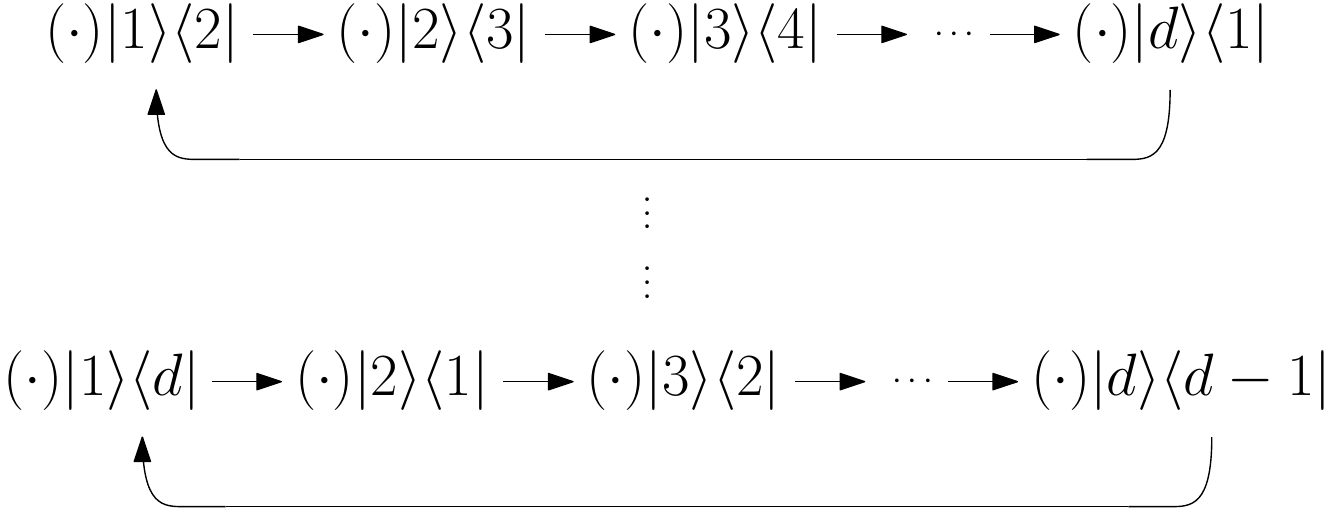}
    \caption{The orbit structure of the off-diagonal matrix units under $\Lambda_+(U)$. The brackets denote constant factors which depend on the entries of $\mathcal{B}$ (see Eq.~\eqref{eq:cycles}).}
    \label{fig:dualcycle-offdiag}
\end{figure}
The eigenvalues arising from each of these cycles can be adjusted to be non-peripheral by choosing all the rows of $C$ to be linearly independent, so that $|\mathcal{B}_{kl}|<1$ for all $k\neq l$. To illustrate why this is true, let us choose the cycle starting from $\ketbra{1}{2}$ and let us denote the eigenvalues that this cycle contributes to $\Lambda_+(U)$ by $\lambda^{(k)}_{12}$. It is easy to check that for all $k$,
\begin{equation}
    [\lambda^{(k)}_{12}]^d = \mathcal{B}_{12}\mathcal{B}_{23}\cdots \mathcal{B}_{d1}.
\end{equation}
By following the same reasoning as in Section~\ref{subsec:non-ergodic}, we can see that the linear independence of the rows of $C$ imply that $|\lambda^{(k)}_{12}|<1$ for all $k$.

\section{Conclusions}
The study of ergodicity of quantum channels is concerned with understanding how a channel behaves under repeated compositions with itself. Since quantum channels model the most general dynamics of a (possibly open) quantum system, analyzing their ergodic properties is crucial in discerning the long-term dynamical behaviour of quantum systems. In this paper, we have studied the ergodic properties of a special class of quantum channels that are covariant with respect to diagonal orthogonal transformations. The action of these channels on input states can be split into a diagonal and an off-diagonal part, with the diagonal action being parameterized by a classical stochastic matrix. We have shown that all the ergodic properties of such channels essentially stem from their `classical cores', i.e., from the aforementioned stochastic matrices. This, in particular, allows us to exploit results from the classical ergodic theory of stochastic matrices to study the ergodic behaviour of the stated covariant quantum channels.

As an application of our analysis, we study brickwork dual unitary quantum circuits which model the time evolution of quantum spin chains.
The long-term behaviour of spatio-temporal correlation functions of local observables provide useful characterizations of ergodicity and mixing in such systems. We consider the case in which the unitary building blocks (i.e., the gates) of the circuit are invariant under local diagonal orthogonal transformations. Under this symmetry, the study of the above-mentioned long term behaviour of correlation functions reduces to the study of ergodic properties of a certain quantum channel which is covariant with respect to diagonal orthogonal transformations. This allows us to exploit our previously obtained results to prove that the stated class of symmetric dual unitary brickwork circuits are diverse enough to exhibit all the possible types of ergodic behaviour as listed previously in \cite{Bruno2019dual}.

\bigskip

\noindent\textbf{Acknowledgements:}
I.N.~was supported by the ANR project ``ESQuisses'' (grant number ANR-20-CE47-0014-01). S.S. gratefully acknowledges support from
the Cambridge Commonwealth, European and International Trust.

\bibliography{references}
\bibliographystyle{alpha}
\bigskip
\hrule
\bigskip

\end{document}